\documentclass[11pt]{article}
\usepackage[margin=1in]{geometry}
\usepackage[round]{natbib}
\newcommand{\backmatter}{}
\newenvironment{keywords}{\par\medskip\noindent\textbf{Key words:}\ }{\par\medskip}

\usepackage[inline,shortlabels]{enumitem}
\usepackage{float}
\usepackage{etoolbox}
\usepackage{mathtools}
\usepackage{graphicx}
\usepackage{tikz}
\usepackage[english]{babel}
\usepackage{longtable}
\usepackage{color}
\usepackage{soul,xcolor}

\usepackage{bbm}
\usepackage{amssymb,amsmath,amsthm}
\usepackage{multirow}
\usepackage[titletoc,title]{appendix}
\usepackage{dsfont}
\usepackage{booktabs}
\usepackage{multirow}
\usepackage{caption}

\usepackage[OT1]{fontenc}
\usepackage{subcaption}
\usepackage{refcount}
\usepackage{array}
\usepackage{accents}
\usepackage{thmtools}
\usepackage{changepage}

\usepackage[colorlinks,citecolor=magenta,urlcolor=magenta]{hyperref}
\newtheorem{proposition}{Proposition}
\newtheorem{theorem}{Theorem}

\theoremstyle{definition}

{
  \theoremstyle{definition}
  \newtheorem{assumption}{Assumption}{}
}
{
  \theoremstyle{definition}
  
}
{
  \theoremstyle{definition}
  \newtheorem{example}{Example}[section]
}

\DeclareMathOperator{\dr}{IF}

\newcommand{\supp}{\mathop{\mathrm{supp}}}

\renewcommand\theassumption{\arabic{assumption}}
\renewcommand\theexample{\arabic{example}}

\newtheorem{lemma}{Lemma} \newtheorem{coro}{Corollary}
\newtheorem{definition}{Definition}

\DeclareMathOperator{\logit}{logit} \DeclareMathOperator{\var}{\mathsf{Var}}
\DeclareMathOperator{\link}{link}
\DeclareMathOperator{\Rem}{Rem} 
\renewcommand{\P}{\mathsf{P}}
\newcommand{\Q}{m}
\newcommand{\rem}{\mathsf{Rem}}
\newcommand{\D}{\varphi}

\renewcommand{\d}{d}

\renewcommand{\r}{g}
\newcommand{\g}{\mathsf{g}}

\newcommand{\indep}{\mbox{$\perp\!\!\!\perp$}} 
 \newcommand{\dd}{\mathrm{d}}

\newcommand{\Pn}{\mathsf{P}_n}

\newcommand{\thetalmtp}{\theta_{\mbox{\scriptsize lmtp}}}

\newcommand{\thetatmle}{\hat\theta_{\mbox{\scriptsize tmle}}}
\newcommand{\thetasr}{\hat\theta_{\mbox{\scriptsize sr}}}
\newcommand{\one}{\mathds{1}}

\newcommand{\E}{\mathsf{E}}

\renewenvironment{proof}{{\it Proof }}{\qed \\}

\DeclarePairedDelimiterX{\norm}[1]{\lVert}{\rVert}{#1}

\definecolor{IvanColor}{HTML}{287A38}   % dark green
\definecolor{NickColor}{HTML}{B35A00}   % dark orange
\definecolor{PawelColor}{HTML}{704090}  % purple
\definecolor{KaraColor}{HTML}{007A80}   % teal

\newcommand{\edit}[1]{#1}
\usepackage{ulem}
\newcommand\redsout{\bgroup\markoverwith{\textcolor{red}{\rule[0.5ex]{2pt}{2.0pt}}}\ULon}
\usetikzlibrary{arrows.meta, positioning, arrows,shapes.arrows,shapes.geometric,shapes.multipart, decorations.pathmorphing,positioning,shapes.swigs}
\newcommand{\titlepaper}{Modified treatment policies that depend on
  the \edit{sequence of natural treatment values.}}

\date{November 3, 2026}
\pdfoutput=1

\newcommand{\biomauthorblock}{%
Iv\'an D\'iaz$^{1,*}$, Nicholas T. Williams$^2$, Pawe{\l} Morzywo{\l}ek$^{3,4}$, and Kara E. Rudolph$^5$\\
$^1$Division of Biostatistics, Department of Population Health, New York University Grossman School of Medicine\\
$^2$Division of Biostatistics, School of Public Health, University of
California at Berkeley\\
$^3$Section of Biostatistics, Department of Public Health, University of Copenhagen\\
$^4$Pioneer Centre for SMARTbiomed, University of Copenhagen\\
$^5$Department of Epidemiology, Mailman School of Public Health, Columbia University\\
$^*$Corresponding author: ivan.diaz@nyu.edu}

\allowdisplaybreaks

\title{\titlepaper}
\author{\begin{minipage}{0.95\textwidth}\centering\small \biomauthorblock\end{minipage}}

\begin{document}
\maketitle
\label{paper:firstpage}
\begin{abstract}
  Longitudinal modified treatment policies (LMTP) are a class of
  interventions that allow the definition, identification, and
  estimation of causal effects in general settings, such as with
  continuous or multivariate exposures, \edit{and} treatment regimens
  that require grace periods.  Targeted machine learning estimators
  (i.e., double/debiased) have been formulated for LMTPs that assign
  the exposure at time $t$ as a function of the natural value of
  treatment at time $t$. However, important applications such as
  estimating the effect of a delay in the start of a treatment,
  \edit{or withholding a potentially harmful treatment,} require
  formulating LMTPs that depend not only on the natural value of
  treatment at time $t$ but also on the \edit{sequence of natural
    values of treatment prior to time $t$}. \edit{This paper develops
    identification and targeted learning methods for estimating mean
    counterfactual outcomes under prespecified LMTPs that may depend
    on the sequence of natural treatment values.} We discuss the
  definition of the effects, and propose estimators that use an
  augmented-data version of the sequential regression form of the
  longitudinal g-computation formula. Our estimators are based on the
  efficient influence function and provide $\sqrt{n}$ inference under
  standard doubly robust rate assumptions on the convergence of the
  outcome and treatment regressions. \edit{We apply the new estimators
    to higher-risk pain treatment to assess the effect of withholding
    treatment the first month it would naturally be indicated on the
    12-month incidence of opioid use disorder.}
\end{abstract}


% Title is typeset by arxiv.tex before the abstract.
%%%%%%%%%%%%%%%%%%%%%%%%%%%%%%%%%%%%%%%%%%%%%%%%%%%%%%%%%%%%%%%%%%%%%%%%%%%%%%% 
\section{Introduction}
%%%%%%%%%%%%%%%%%%%%%%%%%%%%%%%%%%%%%%%%%%%%%%%%%%%%%%%%%%%%%%%%%%%%%%%%%%%%%%% 

Modified treatment policies (MTPs), also known as dynamic
interventions that depend on the natural value of treatment
\citep{Diaz12, haneuse2013estimation, richardson2013single,
  young2014identification, diaz2023nonparametric,
  hoffman2024studying}, provide a flexible way to define, identify,
and estimate causal effects in real-world situations such as
continuous or multivariate exposures \citep{jiang2025ventilator},
limited overlap or practical positivity violations
\citep{wen2023intervention}, grace periods or delays in treatment
initiation \citep{wanis2024grace}, longitudinal studies with
time-varying treatments, censoring, and competing risks
\citep{diaz2024survival}, resource-constrained treatment policies
\citep{sarvet2023longitudinal}, and policy interventions with
spillovers or network interference \citep{balkus2026network}. \cite{sarvet2025natural} provide a thorough review of the literature on methods involving the natural value of treatment.%They
%have been applied to policy-relevant questions where deterministic
%interventions are unrealistic, including shifting mobility patterns
%during COVID-19 \citep{nugent2023mtp_mobility}, modifying service
%intensity in public health nurse home visiting
%\citep{huling2022nurse_home_visiting}, evaluating feasible reductions
%in knee pain to prevent knee replacement
%\citep{jafarzadeh2022pain_knee_replacement}, emulating dental and
%oral-health preservation scenarios
%\citep{cooray2023teeth_social_participation,tay2025masticatory_mortality},
%improving medication adherence \citep{qiu2024metformin_adherence}, and
%assessing policy-relevant environmental shifts such as projected heat
%exposure \citep{barbalat2025heat_language}.

Longitudinal modified treatment policies were originally formulated
broadly enough to allow the intervention at time $t$ to depend on the
\textcolor{black}{sequence of natural values of treatment} up to time $t$
\citep{richardson2013single}. %For times $\ell \leq t$, t
The natural treatment value at time $t$ is the treatment that would
have been observed had the intervention been stopped immediately
before treatment at time $t$ was assigned. Thus, in the original
formulation, the policy at time $t$ may depend both on the
contemporaneous natural treatment value and on earlier natural
treatment values. In contrast, subsequent work on inverse probability
weighted \citep{young2014identification} and targeted learning
estimators \citep{diaz2023nonparametric} focused on the narrower class
of policies in which the intervention at time $t$ depends only on the
contemporaneous natural treatment value. This restriction rules out
policies whose treatment rule depends on the natural values of
treatment at previous timepoints, such as when the intervention
modifies the timing of a naturally occurring treatment process rather
than only the contemporaneous treatment value. For example, in
longitudinal policy evaluation, one may wish to estimate the effect of
delaying the natural date of policy adoption, such as shifting the
enactment time of a naloxone access law \citep{rudolph2022effects}. In
this case, treatment at time $t$ depends on the natural policy status
at an earlier time, or equivalently on the natural adoption
time. %Similarly, consider a dose escalation policy that caps the dose escalation at an increment of at maximum $\delta$ units every time. %would leave the dose unchanged unless the clinician would naturally increase the dose by more than $\delta$ since the previous visit, in which case it would cap the increase at $\delta$.
%Such a dose escalation policy is best represented as a function of the natural dose at time $t$ and at time $t-1$. 

For longitudinal modified treatment policies that depend only on the
contemporaneous natural value of treatment,
\cite{diaz2023nonparametric} developed targeted machine learning
estimators based on \edit{the} efficient influence function that
achieve $\sqrt{n}$-consistency under slow convergence rates for the
nuisance regressions, including the outcome regressions and treatment
mechanisms. These estimators rely on the fact that, for this class of
interventions, the generalized g-computation formula of
\cite{richardson2013single} admits a standard iterated conditional
expectation representation \citep{Bang05,luedtke2017sequential}.  We
show that this representation fails for longitudinal modified
treatment policies that depend on the \edit{sequence of natural values
  of treatment generated at prior time
  points}. %In this setting, the g-computation formula does not correspond to a standard iterated conditional expectation.
We therefore introduce a modified iterated conditional expectation
formula on an augmented data structure and prove that it equals the
corresponding g-computation functional. We then derive the efficient
influence function of the resulting nonparametric functional and
characterize conditions for $\sqrt{n}$-consistent estimation. These
results yield two extensions of the estimators of
\cite{diaz2023nonparametric}: a targeted minimum loss-based estimator
and a sequentially doubly robust estimator. Our estimators are
asymptotically efficient under a \edit{product-rate condition
  requiring products of the relevant treatment- and outcome-regression
  $L_2(\P)$ errors to be $o_P(n^{-1/2})$.} \edit{Our focus is on
  evaluating prespecified policies; learning an optimal policy is
  outside the scope of this work.} We assess finite-sample performance
in simulations and illustrate the methods in an application \edit{to
  higher-risk pain treatment, estimating the effect of withholding the
  first naturally indicated treatment month on the 12-month
  probability of remaining free from opioid use disorder or overdose}
among 12,745 Medicaid-insured patients with lumbar spinal stenosis.

% The rest of the manuscript is organized as follows...

%%%%%%%%%%%%%%%%%%%%%%%%%%%%%%%%%%%%%%%%%%%%%%%%%%%%%%%%%%%%%%%%%%%%%%%%%%%%%%% 
\section{Notation and definition of causal effects}\label{sec:nota}
%%%%%%%%%%%%%%%%%%%%%%%%%%%%%%%%%%%%%%%%%%%%%%%%%%%%%%%%%%%%%%%%%%%%%%%%%%%%%%% 
Let $Z_1,\ldots, Z_n$ denote i.i.d. observations with
$Z=(L_1, A_1, L_2, A_2, \ldots, L_\tau, A_\tau, Y)\sim \P$, where
$L_t$ denotes time-varying covariates, \edit{$A_t$ denotes an exposure
  variable, and $L_{\tau+1}=Y$ denotes an outcome.  We let $\Pn$
  denote the empirical distribution of $Z_1,\ldots,Z_n$, and assume
  $\P$ belongs to the nonparametric statistical model consisting of
  all probability distributions on the sample space of $Z$ that are
  absolutely continuous with respect to a fixed dominating measure
  $\nu$.} We let $\P f = \int f(z)\dd \P(z)$ for a given function
$f(z)$, and let $\E$ denote the expectation with respect to $\P$,
i.e., $\E\{f(Z)\} = \int f(z)\dd\P(z)$. We use $||f||^2$ to denote the
\edit{squared} $L_2(\P)$ norm $\int f^2(z)\dd\P(z)$.
$\bar X_t = (X_1,\ldots, X_t)$ will be used to denote the history of a
variable, $\underline X_t = (X_t,\ldots, X_\tau)$ to denote the future
of a variable, and $H_t = (\bar A_{t-1}, \bar L_t)$ to denote the
history of all variables up until just before $A_t$. When useful, we
will also use notation $\bar X_{t:k} = (X_t,\ldots X_k)$. We let
$\g_t(a_t \mid h_t)$ denote the probability mass function of $A_t$
conditional on $H_t=h_t$, and use calligraphic font to denote the
support of a random variable, e.g., $\mathcal A_t$ denotes the support
of $A_t$. \edit{For each $t$, define
  $\bar{\mathcal A}_t=\prod_{\ell=1}^t\mathcal A_\ell$, the Cartesian
  product of the time-specific treatment supports.} Causal effects
are defined relying on counterfactual logic. Counterfactual random
variables will be formally defined using a non-parametric structural
equation model \citep{Pearl00}. Specifically, for each timepoint $t$,
we assume the existence of deterministic functions $f_{L_t}$,
$f_{A_t}$, and $f_Y$ such that
$L_t=f_{L_t}(A_{t-1}, H_{t-1}, U_{L,t})$, $A_t=f_{A_t}(H_t, U_{A,t})$,
and $Y=f_Y(A_\tau, H_\tau, U_Y)$, where
$U=(U_{L,t}, U_{A,t}, U_Y:t\in \{1,\ldots,\tau\})$ is a vector of
exogenous variables with a distribution satisfying assumptions that
will be given in \S\ref{sec:iden}. \edit{These data-generating
  functions are deterministic but unknown.} Unless defined
otherwise, sums of the type $\sum_{k=\ell}^ta_k$ and products of the
type $\prod_{k=\ell}^ta_k$ for $\ell>t$ are defined as zero and one.

\subsection{Longitudinal modified treatment policies that  only depend on the contemporaneous natural value of treatment}\label{sec:concurrent}

We first revisit the definition of LMTPs used in
\cite{diaz2023nonparametric} to contrast it with the more general
definition we use in this manuscript
(\S\ref{sec:hmtp}). \cite{diaz2023nonparametric} defined LMTP effects
in terms of hypothetical interventions where the equation
$A_t=f_{A_t}(H_t, U_{A,t})$ is removed from the structural model, and
the exposure is assigned as follows.

For a user-defined function $d_1$, at the first timepoint, $t=1$, the
post-intervention treatment is defined as a transformation
$A_1^\d =\d_1(A_1, L_1)$ of the natural value of treatment at that
time, $A_1$, and the history of data collected up until $A_1$,
$L_1$. At time $t=1$, the natural value of treatment is equal to the
observed value of treatment, but this will not be the case for
subsequent timepoints. The hypothetical intervention, $\d_1$,
generates an intervened-on treatment at $t=1$, which we have denoted
$A_1^\d$, and counterfactual $L_2(A_1^\d)$. These are then part of the
counterfactual history and the natural value of treatment at time
$t=2$, denoted with $H_2(A_1^\d) = (L_1, A_1^\d, L_2(A_1^\d))$ and
$A_2(A_1^\d)$.

At time $t=2$, the post-intervention treatment is defined as
$A_2^\d =\d_2(A_2(A_1^\d), H_2(A_1^\d))$. In particular, note that
this intervention depends only on the \textit{contemporaneous} natural
value of treatment at time $t=2$, $A_2(A_1^d)$, and not on the natural
value of treatment at $t=1$, $A_1$. Instead, the intervention depends
on the intervened-on value of treatment at $t=1$, namely
$A_1^d=\d_1(A_1, L_1)$. This intervention at time $t=2$ generates an
intervened-on $A_2^d$ and updated counterfactual history, which are
then used to assign the post-intervention exposure at time $t=3$ as
$A_3^\d =\d_3(A_3(\bar A_2^\d), H_3(\bar A_2^\d))$. This process is
iterated forward in time until interventions at all times $t\in[\tau]$
have been performed, at which point we obtain the counterfactual
outcome
$Y(\bar A^\d)=f_Y(A_\tau^\d, H_\tau(\bar A_{\tau-1}^\d),
U_Y)$. \edit{For a generic history $h_t=(\bar a_{t-1},\bar l_t)$, we
  write the policy as $d_t(s_t,h_t)$, where $s_t$ denotes the current
  natural treatment value and $\bar a_{t-1}$ denotes the treatment
  values appearing in the history. These arguments are evaluated at
  $A_t(\bar A_{t-1}^\d)$ and $H_t(\bar A_{t-1}^\d)$, respectively. The
  formal definition is as follows.}
\begin{definition}\label{def:lmtp-current}
  An intervention variable $A_t^\d$ is said to be a
  \textit{longitudinal modified treatment that depends on the
    contemporaneous natural value of treatment} if, for a user-given
  function $\d_t$, it is defined as
  $ A^\d_t = \d_t\{A_t(\bar A_{t-1}^\d),H_t(\bar A_{t-1}^\d)\}$.
\end{definition}

%One can also allow the function $d_t$ to depend on a randomizer that
%is independent of $U$ and does not depend on $\P$, but we omit that
%case in this paper for simplicity.
The data generating process that leads to the above counterfactual
outcome is illustrated in Figure \ref{fig:swig-current} with $\tau=2$
and omitting $L_1$. \cite{hoffman2024studying} provide a large list of
interventions of practical interest that can be described using this
framework. Importantly, delay interventions do not fit
this definition. We elaborate on this point in \S\ref{sec:delay}.

\subsection{Longitudinal modified treatment policies that depend on
  the \edit{sequence of past natural values of treatment}}\label{sec:hmtp}

In this paper, we generalize the above definition to allow dependence
of the intervention at time $t$ on the \edit{sequence of
  prior natural values of treatment}---instead of just on the
contemporaneous natural value of treatment. Specifically,
counterfactual outcomes are generated as follows.

The intervention at $t=1$ is defined identical to
\S\ref{sec:concurrent}. At time $t=1$, the post-intervention variable
is defined as a transformation $A_1^\d =\d_1(A_1, L_1)$ of the natural
value (which, again, is also the observed value) of treatment $A_1$
and $L_1$. This intervention generates a counterfactual history and
treatment at time $t=2$, denoted with
$H_2(A_1^\d) = (L_1, A_1^\d, L_2(A_1^\d))$ and $A_2(A_1^\d)$,
respectively.

At time $t=2$, the two definitions diverge. At this timepoint the
post-intervention exposure is generated as
$A_2^\d =\d_2(\bar A_2(A_1^\d), H_2(A_1^\d))$. The difference between
this and the definition of $A_2^\d$ in \S\ref{sec:concurrent} is the
overbar for the \edit{sequence of past natural values
  $\bar A_2(A_1^\d)$}, which is defined as
$\bar A_2(A_1^\d)=(A_1, A_2(A_1^\d))$. In particular, note that this
intervention can depend on two values of treatment at $t=1$, the
natural value of treatment, $A_1$, and the intervened value of
treatment, $A_1^d=\d_1(A_1,L_1)$. This intervention at time $t=2$
generates a counterfactual history and natural value of treatment at
time $t=3$,
$H_3(\bar A_2^\d) = (L_1, A_1^\d, L_2(A_1^\d), A_2^\d, L_3(\bar
A_2^\d))$ and $A_3(\bar A_2^\d)$, similar to \S\ref{sec:concurrent}.
These variables are then used to assign the post-intervention exposure
at time $t=3$ as
$A_3^\d =\d_3(\bar A_3(\bar A_2^\d), H_3(\bar A_2^\d))$, where
$\bar A_3(\bar A_2^\d)=(A_1, A_2(A_1^\d), A_3(\bar A_2^d))$. This
process is iterated forward in time until interventions at all times
$t\in[\tau]$ have been performed, at which point we obtain the
counterfactual outcome
$Y(\bar A^\d)=f_Y(A_\tau^\d, H_\tau(\bar A_{\tau-1}^\d), U_Y)$. Below
is the formal definition, which coincides with the general definition
given in \S5.2 of \cite{richardson2013single}. Figure~\ref{fig:swig}
gives a graphical illustration of the differences between this
definition and the definition of
\S\ref{sec:concurrent}. %used in the construction of previous targeted learning estimators.

\edit{For each \(\ell\leq t\), \(A_\ell(\overline A_{\ell-1}^d)\) is the
  treatment value generated by the natural treatment mechanism after
  the policy has been applied through time \(\ell-1\), but before the
  policy modifies treatment at time \(\ell\).  This value is then mapped
  to \(A_\ell^d\), and the recursion proceeds to time \(\ell+1\). Thus, the
  components of \(\overline A_t(\overline A_{t-1}^d)\) are the
  successive pre-intervention treatment values generated during
  implementation of a single recursively defined policy. We write the
  generalized policy as $d_t(\bar s_t,h_t)$, where
  $\bar s_t=(s_1,\ldots,s_t)$ denotes a natural treatment history. For
  $j<t$, the natural treatment argument $s_j$ is distinct from the
  treatment coordinate $a_j$ in $h_t$. Under the intervention, these
  are evaluated at $A_j(\bar A_{j-1}^\d)$ and $A_j^\d$, respectively,
  and may therefore differ.}
\begin{definition}\label{def:glmtp}
  An intervention $A_t^\d$ is a \textit{generalized longitudinal
    modified treatment policy that depends on the \edit{sequence of
      natural values of treatment}} (G-LMTP) if, for a user-given
  function $\d_t$, it is defined as
  \edit{$ A^\d_t=\d_t\{\bar A_t(\bar A_{t-1}^\d),H_t(\bar
    A_{t-1}^\d)\}, $} where the \edit{sequence of natural values of
    treatment} is
  \edit{$\bar A_t(\bar A_{t-1}^\d)
      =\{A_1,A_2(A_1^\d),\ldots,A_t(\bar A_{t-1}^\d)\}.$}
\end{definition}
\edit{For a prespecified policy $\d=(\d_1,\ldots,\d_\tau)$, our
  parameter of interest is the mean counterfactual outcome
  $ \thetalmtp = \E\{Y(\bar A^\d)\}$.} \edit{In the next subsection we
  discuss scientifically and clinically relevant interventions that
  require the G-LMTP framework}, highlighting issues that arise when
more narrowly defined %standard
LMTPs do not have access to the \edit{sequence of past natural
  treatment values} once a modification has occurred. \edit{Table
  \ref{tab:comparison} provides a comparison of the notation used in
  the definitions of LMTPs vs G-LMTPs.}

\begin{table}[htbp]
\centering
\caption{\edit{Comparison of notation in the intervention definitions
for LMTPs and G-LMTPs.}}
\label{tab:comparison}

\renewcommand{\arraystretch}{1.7}
\small

\begin{tabular}{@{}
    >{\raggedright\arraybackslash}p{4cm}
    >{\raggedright\arraybackslash}p{5cm}
    >{\raggedright\arraybackslash}p{5cm}
    @{}}
\toprule
\textbf{Quantity}
& \textbf{Section 2.1: LMTP}
& \textbf{Section 2.2: G-LMTP} \\
\midrule

Intervention function
& \edit{$d_t(s_t,h_t)$}
& \edit{$d_t(\bar s_t,h_t)$} \\

Intervention at $t=1$
& $A_1^d=d_1(A_1,L_1)$
& $A_1^d=d_1(A_1,L_1)$ \\

Natural treatment at $t=2$
& $A_2(A_1^d)$
& $A_2(A_1^d)$ \\

Natural treatment history at $t=2$
& Not used explicitly
& $\bar A_2(A_1^d)=(A_1,A_2(A_1^d))$ \\

Intervention at $t=2$
& $
A_2^d=d_2(A_2(A_1^d),
          H_2(A_1^d))$
& $
A_2^d=d_2(\bar A_2(A_1^d),
         H_2(A_1^d))$ \\

Natural treatment at $t=3$
& $A_3(\bar A_2^d)$
& $A_3(\bar A_2^d)$ \\

Natural treatment history at $t=3$
& Not used explicitly
& $\bar A_3(\bar A_2^d)$ \\

Intervention at $t=3$
& $
A_3^d=d_3(A_3(\bar A_2^d),
         H_3(\bar A_2^d))$
& $
A_3^d=d_3(\bar A_3(\bar A_2^d),
         H_3(\bar A_2^d))$ \\

Counterfactual outcome
& $Y(\bar A^d)$
& $Y(\bar A^d)$ \\

\bottomrule
\end{tabular}
\end{table}

% For consistency, replace the two occurrences of \thetalmtp in
% Section 3 with \theta.
\begin{figure}[H]
  \centering
  \begin{subfigure}[b]{0.9\textwidth}
    \centering
    \begin{tikzpicture}
      \tikzset{line width=1.5pt, outer sep=0pt,
        ell/.style={draw,fill=white, inner sep=2pt,
          line width=1.5pt},
        swig vsplit={gap=5pt,
          inner line width right=0.5pt}}
      \node[name=a1, shape=swig vsplit, swig vsplit={
        gap=15pt}]{
        \nodepart{left}{$A_1$}
        \nodepart{right}{$A_1^\d$} };
      \node[name=l2, right=5mm of a1, ell, shape=ellipse]{$L_2(A_1^\d)$};
      \node[name=a2,shape=swig vsplit,swig vsplit={
        gap=15pt}, right=5mm of l2]{
        \nodepart{left}{$A_2(A_1^\d)$}
        \nodepart{right}{$A_2^\d$} };
      \node[name=y, right=5mm of a2, ell, shape=ellipse]{$Y(\bar A^\d)$};
      \draw[->,line width=1.5pt,>=stealth]
      (a1.left center) edge (a1.right center)
      (a2.left center) to (a2.right center)
      (a1) edge[out=330,in=200] (a2.320)
      (a1) edge (l2)
      (l2) edge[out=30,in=150] (a2.20)
      (a1) edge[out=15,in=150] (a2)
      (a1) edge[out=30,in=150] (y)
      (l2) edge[out=330,in=210] (y)
      (l2) edge (a2)
      (a2) edge (y);
    \end{tikzpicture}
    \caption{LMTP that depends on the natural \textit{contemporaneous} value of 
      treatment}
    \label{fig:swig-current}
  \end{subfigure}

  \begin{subfigure}[b]{0.9\textwidth}
    \centering
    \begin{tikzpicture}
      \tikzset{line width=1.5pt, outer sep=0pt,
        ell/.style={draw,fill=white, inner sep=2pt,
          line width=1.5pt},
        swig vsplit={gap=5pt,
          inner line width right=0.5pt}}
      \node[name=a1, shape=swig vsplit, swig vsplit={
        gap=15pt}]{
        \nodepart{left}{$A_1$}
        \nodepart{right}{$A_1^\d$} };
      \node[name=l2, right=5mm of a1, ell, shape=ellipse]{$L_2(A_1^\d)$};
      \node[name=a2,shape=swig vsplit,swig vsplit={
        gap=15pt}, right=5mm of l2]{
        \nodepart{left}{$A_2(A_1^\d)$}
        \nodepart{right}{$A_2^\d$} };
      \node[name=y, right=5mm of a2, ell, shape=ellipse]{$Y(\d)$};
      \draw[->,line width=1.5pt,>=stealth]
      (a1.left center) edge (a1.right center)
      (a1.95) edge[out=30, in=160, dashed] (a2.40)
      (a2.left center) to (a2.right center)
      (a1) edge[out=330,in=200] (a2.320)
      (a1) edge (l2)
      (l2) edge[out=30,in=150] (a2.20)
      (a1) edge[out=15,in=150] (a2)
      (a1) edge[out=30,in=150] (y)
      (l2) edge[out=330,in=210] (y)
      (l2) edge (a2)
      (a2) edge (y);
    \end{tikzpicture}
    \caption{LMTP that depends on the \edit{sequence of past natural
        values of treatment}. The dashed line is used to highlight the
      additional arrow used here, as compared to the figure above.}
    \label{fig:swig-history}
  \end{subfigure}
  \caption{Single world intervention graphs (SWIGs) for a modified
    treatment policy with two timepoints. The dashed arrow
    corresponds to the generalization addressed in this work.}
    \label{fig:swig}
\end{figure}

%\subsection{Interventions that depend on the \edit{sequence of past
%    natural values of treatment}}

\subsection{Delay interventions}\label{sec:delay}
\edit{It is often of interest to evaluate the effect of a delay in the
  start of an exposure, intervention, or policy. Consider first a
  binary lag intervention to show why a policy that shifts a natural
  treatment value forward in time requires the G-LMTP framework. For
  example,} in vaccine studies, one may be interested in assessing the
effect of a delay of $k$ timepoints in receiving a vaccine
booster. \edit{If $A_t$ denotes receipt of a booster at time $t$, this
  intervention may be represented as $d_t(\bar s_t, h_t)=0$ if
  $t\leq k$ and $d_t(\bar s_t, h_t)=s_{t-k}$ otherwise. Alternatively,
  if the exposure is an absorving state (e.g., having received a
  booster by or at time $t$), the formulation of a delay policy
  requires different considerations; see Supplementary Materials
  \S\ref{app:example}.}
% Here we briefly show that such delay policies cannot be expressed
% within the LMTP framework described by \cite{diaz2023nonparametric},
% but can be expressed by the G-LMTP framework described in this
% paper. %Notably, the definition and identification (but not the efficient estimation) of effects that depend on the history of the natural value of treatment had already been discussed by \cite{richardson2013single}, albeit using a different causal model.  Consider a simple setting with only two time points and a binary
% intervention at each timepoint.
Consider for illustration a simple setting with only two time points.
A delay in exposure of one timepoint would correspond to an
intervention that assigns treatment at time $t=1$ as
$A_1^\d = \d_1(A_1, L_1) = 0$. At time $t=2$, the desired delay
intervention should assign the treatment as $A_2^\d =
A_1$. %This, however, is not allowed
%by the SWIG in 
This however, is not possible in the setting depicted by Figure
\ref{fig:swig-current}, as there is no dependence of the MTP on the
prior natural value of treatment %is not allowed
under that setup. In other words,
$A_2^\d = \d_2(A_2(A_1^\d), L_2(A_1^\d), A_1^\d, L_1)$ must be defined
solely in terms of the natural value of treatment at time 2 and the
\textit{intervened} history, i.e., $d_1(A_1, L_1)$. If $d_1$ modifies
$A_1$ (e.g., from 1 to 0), the function $d_2$ no longer has access to
the original $A_1$, and thus cannot shift it forward in time. In
contrast, under \edit{a} G-LMTP, $A_2^\d$ may be defined as a function of
$A_1$ even if $A_1^\d \ne A_1$, and the policy can be explicitly
constructed to reassign the original treatment at time 1 to occur at
time 2 instead.

\edit{Similarly, consider an intervention withholding a potentially
  harmful intervention such as invasive mechanical ventilation in
  acutely ill patients.}
%To further elaborate, consider
%a long-standing question in critical care regarding the effects of
%delaying invasive mechanical ventilation.
Let $A_t \in \{0,1,2\}$ denote the level of oxygen support at time
$t$, where $A_t = 0$ indicates no support, $A_t = 1$ indicates
non-invasive support, and $A_t = 2$ indicates invasive mechanical
ventilation. For an arbitrary vector \edit{$\bar s_t$}, define the
function%---which we will redefine using G-LMTP notation shortly---as
\[
  \edit{d_t(\bar s_t,h_t)} =
  \begin{cases}
    1 & \text{if } s_t = 2 \text{ and } s_{\ell} \leq 1 \text{ for all } \ell < t, \\
    s_t & \text{otherwise,}
  \end{cases}
\]
%An G-LMTP defined by
%$A_t^\d = d_t(A_t(\bar A_{t-1}^\d), \bar A_{t-1}(\bar A_{t-2}^\d))$
which represents a policy that \edit{withholds} mechanical ventilation
by substituting it with non-invasive support \edit{at the first
  natural indication of invasive mechanical ventilation (IMV)}, where
the G-LMTP is defined by
\edit{$A_t^\d=d_t\{\bar A_t(\bar A_{t-1}^\d), H_t(\bar
  A_{t-1}^\d)\}$.} Consider $\tau=3$ and a patient with observed
treatment history $(A_1, A_2, A_3) = (1, 2, 2)$. At time $t = 1$, the
patient would receive $A_1^\d = 1$. At time $t = 2$, since the natural
treatment is $A_2(1) = 2$ and the patient has not yet \edit{had a
  prior natural indication of invasive ventilation}, the policy
prescribes $A_2^\d = 1$. At time $t = 3$, \edit{because \(A_2(1)=2\),
  the condition that all previous natural treatment values are at most
  1 is false. The policy therefore assigns \(A_3^d=A_3(1,1)\)}. That
is, the possible paths under this intervention are $(1, 1, 2)$,
$(1, 1, 1)$, and $(1, 1, 0)$, depending on the natural value of
\edit{$A_3(1,1)$}. \edit{This intervention cannot generally be
  represented by a current-time LMTP. At time \(t=3\), the rule must
  determine whether invasive mechanical ventilation was naturally
  indicated and delayed at time \(t=2\). A current-time LMTP cannot
  make this determination because the modified treatment history does
  not retain \(A_2(A_1^d)\).}

\section{Identification of causal effects}\label{sec:iden}

\edit{For fixed natural treatment values $\bar s_t$ and covariate
  values $\bar l_t$, we recursively define
  $a_t^\d=\d_t(\bar s_t,h_t^\d)$, where
  $h_t^\d=(\bar a_{t-1}^\d,\bar l_t)$.  Here, $a_t^\d$ denotes the
  treatment assigned by the policy at time $t$, and $h_t^\d$ contains
  the previously assigned treatment values and covariate history.}
G-LMTPs will be identified under the following
assumptions. Alternative assumptions \edit{using counterfactuals} were
given by \cite{richardson2013single} and
\cite{young2014identification}.
\begin{assumption}[Positivity]\label{ass:support}
  \edit{For each $t\in\{1,\ldots,\tau\}$ and every
  $(\bar s_t,\bar l_t)$ satisfying
  $(s_k,l_k)\in\supp\{A_k,L_k\mid A_{k-1}=a_{k-1}^d,
  H_{k-1}=h_{k-1}^d\}$ for $k=1,\ldots,t$, we have
  $(a_t^d,h_t^d)\in\supp\{A_t,H_t\}$. The condition is applied
  recursively, with the support at $k=1$ interpreted as
  $\supp\{A_1,L_1\}$.}
\end{assumption}
\begin{assumption}[Strong sequential randomization]\label{ass:exch}
  $U_{A,t}\indep (\underline U_{L, {t+1}}, \underline U_{A, t+1}, U_Y) \mid H_t$ for all  $t\in\{1,\ldots,\tau\}$.
\end{assumption}
We have the following identification theorem, which allows us to
compute the parameters $\thetalmtp$ as a function of the observed data
distribution.
% Lemma 1 in the
% Appendix to \cite{kennedy2018nonparametric} \citep[see also the
% g-formula of][]{Robins86}.
\edit{Strong sequential randomization allows counterfactual
  conditional distributions to be replaced by their observed-data
  counterparts. Positivity ensures that the intervened histories
  required by the policy are supported by the observed data.}
\begin{theorem}[Identification of the effect of G-LMTPs]\label{theo:iden}
  \item Under Assumptions \ref{ass:support} and \ref{ass:exch},
    \edit{$\thetalmtp = \E\{Y(\bar A^\d)\}$} is identified as
    \[
\theta
=
\int_{\bar{\mathcal A}_\tau,\bar{\mathcal L}_\tau}
\E[Y\mid A_\tau=a_\tau^\d,H_\tau=h_\tau^\d]
\prod_{k=1}^\tau
\dd\P(\edit{s_k},l_k\mid
A_{k-1}=a_{k-1}^\d,H_{k-1}=h_{k-1}^\d).
\]
\end{theorem}
\edit{Intuitively, the identifying functional is obtained by following
  the G-LMTP forward in time. For a sequence of natural treatment
  values $\bar s_\tau$, the policy recursively maps each $s_t$ into
  the corresponding intervened value $a_t^\d$. The conditional
  distribution of $(A_t,L_t)$ given the intervened history describes
  the distribution of the next natural treatment value and covariates
  under this intervention. Multiplying these conditional distributions
  across time therefore gives the distribution of the covariate and
  natural treatment trajectory generated while the policy is
  implemented.  Finally, $\E[Y\mid A_\tau=a_\tau^\d,H_\tau=h_\tau^\d]$
  gives the conditional mean outcome under the resulting intervened
  history.  Integrating over all possible natural treatment and
  covariate trajectories yields $\E\{Y(\bar A^\d)\}$.}

The proof of this theorem is identical to the first part of the proof
of Theorem~1 in \cite{diaz2023nonparametric} and is given in the
Supplementary Materials. Theorem~1 in \cite{diaz2023nonparametric}
additionally shows that current-time LMTPs admit a sequential
regression representation of the identification formula, which greatly
facilitates their estimation because it allows for the straightforward
generalization of estimators developed for static interventions
\citep{luedtke2017sequential, rotnitzky2017multiply}. In contrast,
constructing a sequential regression representation for G-LMTPs
requires more care. Consider a simple example with only two
timepoints, no baseline covariates, and an intervention that only
depends on the \edit{sequence of natural treatment values}, i.e.,
$A_1^d=d_1(A_1)$ and $A_2^d = d_2(A_1, A_2(A_1^\d))$. The outcome
expectation under this intervention is identified as
\[\int E[Y\mid A_2=\edit{d_2(s_1,s_2)}, L_2=l_2, A_1=a_1^d]\dd\P(s_2,
  l_2\mid A_1=a_1^d)\dd\P(s_1),\] where we recall that
$a_1^d=d(s_1)$. The integral with respect to $(s_2,l_2)$ cannot be
written as a conditional expectation because the outcome regression
involves setting $A_1$ to two different values $s_1$ and
$a_1^d$. However, progress can be made by decoupling the value of
$s_1$ used in the conditioning set $\P(s_2, l_2\mid A_1=d_1(s_1))$
from that used in $\edit{d_2(s_1,s_2)}$ (see \S\ref{sec:estima} and Figure
\ref{fig:pooled} for more discussion). Specifically, define
$ m_2(a_2, h_2)=E[Y\mid A_2=a_2, L_2=l_2, A_1=a_1]$. Then, for each
$s_1$ in the support of $A_1$, define
$m_1(s_1,a_1)=E[m_2(d_2(s_1, A_2), H_2)\mid A_1=a_1]$, and note that
$\theta=E[m_1(A_1, d_1(A_1))]$. This can be generalized as in the
following proposition.

\begin{proposition}[Sequential regression representation]\label{propo:sr}
  Let
  $m_{\tau}(a_\tau, h_\tau)=E[Y\mid A_\tau=a_\tau, H_\tau=h_\tau]$,
  and
  $q_\tau(\bar s_{\tau-1}, A_{\tau}, H_{\tau}) = m_\tau(d_\tau((\bar
  s_{\tau-1}, A_\tau), H_\tau), H_\tau)$. \edit{For each $\bar s_t$},
  recursively for $t=\tau-1,\ldots, 1$ define
  \begin{align*}
    m_t(\bar s_t,a_t, h_t)&=E[q_{t+1}(\bar s_t, A_{t+1}, H_{t+1})\mid A_t=a_t, H_t = h_t]\\
        q_t(\bar s_{t-1}, a_t, h_t) &= m_t((\bar s_{t-1}, a_t), d_t((\bar s_{t-1},
  a_t), h_t),
    h_t),
  \end{align*}    
  where we define $s_0=\emptyset$. Then we have $\theta = E[q_1(A_1,
  H_1)]$. 
\end{proposition}
\edit{Note that the functions $m_t$ and $q_t$ are only required on
  values of their arguments encountered in the recursive evaluation of
  the identifying functional. Under Assumption~1, these evaluations
  occur within the support of the observed-data distribution.} To
illustrate this proposition, we consider an example with $\tau=3$ in
the Supplementary Materials.
\section{Optimality theory}
Efficiency theory in this paper focuses on functions $d$ that are
known and need not be estimated (i.e., do not depend on
$\P$). \edit{For the efficiency theory and estimators developed in the
  remainder of the paper, we assume that $A_t$ is discrete.} For each
$t< k$, and each $s_t$ define the random variable
$D_{t,k}(\bar s_k)=\one\{A_{t+1}=d_{t+1}(\bar s_{t+1}, H_{t+1}),
\ldots A_k=d_k(\bar s_k, H_k)\}$. For $t=\tau-1,\ldots,0$, define
$\bar\varphi_{t+1}(\bar s_t, Z) = \varphi_{t+1}(\bar s_t, Z) -
m_t(\bar s_t, A_t, H_t)$, where $m_0=\theta$. \edit{The centering
regression $m_t$, including $m_0=\theta$, is held fixed when the nuisance
argument of $\bar\varphi_{t+1}$ is varied.} Define 
% \begin{multline}
%   \label{eq:eifgamma}
%  \D_1 : Z\mapsto
%   \sum_{\ell=1}^\tau\sum_{\bar a_\ell\in D_{1,\ell}}\left(\prod_{k=1}^\ell \frac{g_\ell(a_\ell\mid H_\ell)}{g_\ell(A_\ell\mid
%       H_\ell)}\right)\times \\\{m_{\ell+1}(A_{\ell+1}, d(\bar a_\ell, A_{\ell+1}, H_{\ell+1}), H_{\ell+1}) -
%   m_\ell(a_\ell, A_\ell, H_\ell)\}\\ + m_1(A_1, d(A_1, H_1), H_1) - \theta
% \end{multline}
% \begin{multline}
%   \label{eq:eifgamma}
%   \D_t : Z\mapsto
%   \sum_{k=t}^\tau\sum_{\bar a_k\in D_{t,k}}\left(\prod_{u=t}^k \frac{g_u(a_u\mid H_u)}{g_u(A_u\mid
%       H_u)}\right)\times \\\{q_{k+1}(a_k, A_{k+1}, H_{k+1}) -
%   m_\ell(a_k, A_k, H_k)\}\\ + q_t(A_{t-1}, A_t, H_t)
% \end{multline}
{\small\begin{equation}
    \label{eq:eifgamma}
    \begin{aligned}
  \varphi_{t+1}(\bar s_t, Z)&=
           \sum_{k=t+1}^\tau\sum_{\bar s_{(t+1):k}}D_{t,k}(\bar s_k)\left(\prod_{u=t+1}^k
           \edit{\frac{g_u(s_u\mid H_u)}{g_u(A_u\mid H_u)}}\right)\{q_{k+1}(\bar s_k, A_{k+1}, H_{k+1}) -
           m_k(\bar s_k, A_k, H_k)\}\\
      &+ q_{t+1}(\bar s_t, A_{t+1}, H_{t+1}).
      \end{aligned}
    \end{equation}}%
  
 When necessary, we use the notation
  $\D_t(\cdot;\eta)$ or $\D_t(\cdot;\underline\eta_t)$ to highlight the
  dependence of $\D_t$ (or $\bar\D_t$) on
  $\underline\eta_t=(\r_t,\Q_t,\ldots,\r_\tau, \Q_\tau)$. We also use $\eta$ to denote
  $(\r_1,\Q_1,\ldots,\r_\tau, \Q_\tau)$, and define
  \edit{$\D_{\tau+1}(\bar s_\tau,Z;\eta)=Y$ and
  $q_{\tau+1}(\bar s_\tau,A_{\tau+1},H_{\tau+1})=Y$, with the latter
  notation used only as a terminal convention.} 
% [Similar to
% https://www.degruyter.com/downloadpdf/j/ijb.2012.8.issue-1/1557-4679.1370/1557-4679.1370.pdf and https://arxiv.org/pdf/1705.02459.pdf]
  In the following, we let
  $\eta'=(\r_1',\Q_1',\ldots,\r_\tau', \Q_\tau')$ denote some value of
  $\eta$. This value will typically represent the probability limit of a
  given estimator $\hat\eta$. The function $\D_t$ satisfies
  the following property, which will be crucial to establish consistency
  of some estimators of $\theta$ under \edit{sequential double robustness} assumptions,
  and to construct estimators of $\theta$ under slow convergence rates
  for estimation of $g_t$ and $m_t$.

\begin{theorem}[Conditional von Mises expansion]\label{theo:sdr}
    Define 
    $C_{t,k}' = D_{t,k}(\bar s_k)\prod_{r=t+1}^{k-1}\edit{\frac{g_r'(s_r\mid H_r)}{g_r'(A_r\mid H_r)}}$.
    For each $t\in\{0,\ldots,\tau-1\}$, and for any $\eta'$, define the second order error
    term 
    {\footnotesize
      \begin{multline}
        \rem_t(\bar s_t, a_t, h_t;\eta')=\\
        \sum_{\substack{k>t\\\bar s_{(t+1):k}}}\E\left[C_{t,k}'\left\{\frac{g_k'(s_k\mid
              H_k)}{g_k'(A_k\mid H_k)}-\frac{g_k(s_k\mid
              H_k)}{g_k(A_k\mid H_k)}\right\}\{\Q_k'(\bar s_k, A_k, H_k) -
          \Q_k(\bar s_k, A_k, H_k)\}\,\,\bigg|\,\, A_t=a_t,H_t=h_t \right],\label{eq:rem}
      \end{multline}
    }%
    where \edit{there is no conditioning at $t=0$}, and for
    $t=\tau$ we let $\rem_\tau(\bar s_\tau, a_\tau, h_\tau;\eta')=0$. The
    \edit{regression $\Q_t$ admits the following conditional expansion:
    \begin{equation}
      \begin{aligned}
      \Q_t(\bar s_t,a_t,h_t)
      ={}&\E\{\D_{t+1}(\bar s_t,Z;\eta')\mid A_t=a_t,H_t=h_t\}+\rem_t(\bar s_t,a_t,h_t;\eta').
      \end{aligned}\label{eq:fo}
    \end{equation}}%
  \end{theorem}
  This theorem shares important connections to the von Mises-type
  expansions used in some of the semi-parametric inference literature
  \citep[e.g.,][]{mises1947asymptotic, vanderVaart98,
    robins2009quadratic}. Inspection of this theorem teaches us that \edit{the
  transformation $\D_{t+1}(\bar s_t,Z;\eta')$ has conditional
  mean error $-\rem_t(\bar s_t,a_t,h_t;\eta')$ relative to $\Q_t$.
  Regressing this transformation on $(A_t,H_t)$ for each $\bar s_t$
  therefore motivates a sequential doubly robust regression estimator.} Furthermore, this implies the following corollary:
  
\begin{coro}[Efficient Influence Function]\label{coro:eif}
  Assume that $\d$ does not depend on $\P$. The efficient influence
  function for estimating $\theta = \E[q_1(A_1, H_1)]$ in the
  non-parametric model is given by $\bar\D_1(Z)$.
\end{coro}
%%%%%%%%%%%%%%%%%%%%%%%%%%%%%%%%%%%%%%%%%%%%%%%%%%%%%%%%%%%%%%%%%%%%%%%%%%%%%%% 

% Intuitively, sequential double robustness of $\D_1$ occurs
% because Assumption\ref{ass:inv} allows us to use the change of
% variable formula to change measures using the Radon-Nikodym derivative
% $\r_t$. Specifically, we have two equivalent representations of
% $\Q_t$:
% \begin{align*}
    %     \Q_t(a_t, h_t)&=\E[\Q_{t+1}(A^\d_{t+1}, H_{t+1})\mid A_t=a_t,
                          %                           H_t=h_t]\\
    %                         &=\E[\Q_{t+2}(A_{t+2}, H_{t+2})\r_{t+1}(A_{t+1}, H_{t+1})\mid A_t=a_t, H_t=h_t]
                                %   \end{align*}

%\begin{remark}
If the modified treatment policy depends only on the contemporaneous
natural value of treatment, i.e., if
$\edit{d_t(\bar s_t,h_t)=d_t(s_t,h_t)}$, then all the above results
reduce to the results presented by \cite{diaz2023nonparametric}. If
the intervention only depends on some of the \edit{sequence of natural
  treatment values}, then the argument $\bar s_{t-1}$ in the function
$q_t$ can be modified to only include values on which the intervention
depends. We elaborate on this point in the Supplementary Materials
\S\ref{app:example}.
%\end{remark}

\section{Estimation and statistical inference}\label{sec:estima}

In this section we assume that preliminary estimators $\hat\r_t$, and
$\hat\Q_t$ are available. These estimators may be obtained from
flexible regression methods from the machine learning literature or
ensembles thereof \citep{vanderLaan&Polley&Hubbard07}. As
implied by Theorem~\ref{theo:sdr}, the consistency of these estimators
will determine the consistency of our estimators of the parameter
$\theta$. We will use Super Learner \citep{vanderLaan&Polley&Hubbard07}, an ensemble regression algorithm
\edit{that builds} a convex combination of regression algorithms in a
user-given library, with weights chosen to minimize the
cross-validated prediction error. 

\edit{A plug-in estimator follows from Proposition~\ref{propo:sr} by
  estimating the sequential regressions and averaging the resulting
  predictions $\hat q_1(A_{1,i},H_{1,i})$ over individuals. For
  treatments with finite support, these regressions can be fitted
  using pooled regressions in augmented datasets. Computational
  details and a worked example are provided in the Supplementary
  Materials, \S\ref{app:pooled-implementation}. Although this na\"ive
  estimator is simple to implement, its error is first order in the
  errors of the sequential regression estimates.  Consequently, under
  data-adaptive nuisance estimation, it is not generally
  $\sqrt{n}$-consistent or asymptotically linear under the
  product-rate conditions that suffice for the TMLE and SDR
  estimators.  It therefore does not provide the same general basis
  for valid Wald-type confidence intervals under these conditions.} To
achieve this, we will leverage two approaches that rely on the
canonical gradient in \edit{Corollary~\ref{coro:eif}}: targeted
minimum loss-based estimation (TMLE), and sequential doubly robust
(SDR) regression. Both of these estimators are $\sqrt{n}$-consistent,
efficient, and asymptotically normal under assumptions, but the SDR
has the advantage over TMLE that it has a simpler error structure
therefore leading to potential gains in finite sample bias and
consistency, whereas the advantage of the TMLE over SDR is that it is
a substitution estimator therefore offering protection in cases where
the parameter space is bounded and the true value of the parameter is
close to the boundary. The discussion about the difference between
TMLE and SDR estimators is given elsewhere
\citep[e.g.,][]{molina2017multiple, luedtke2017sequential,
  diaz2023nonparametric} and is therefore omitted here.

\subsection{Targeted minimum loss-based estimator}\label{sec:tmle}

We use sample splitting to avoid imposing Donsker conditions on the
\edit{initial} regression estimators \citep{klaassen1987consistent,zheng2011cross,
  chernozhukov2018double}. Let ${\cal V}_1, \ldots, {\cal V}_J$ denote
a random partition of the index set $\{1, \ldots, n\}$ into $J$
prediction sets of approximately the same size. For each $j$, the
associated training sample is given by
${\cal T}_j = \{1, \ldots, n\} \setminus {\cal V}_j$. We let
$\hat \eta_{j}$ denote the estimator of $\eta$ obtained by training
the corresponding prediction algorithm using only data in the sample
${\cal T}_j$. Further, we let $j(i)$ denote the index of the
validation set which contains observation $i$.
% \edit{The targeted fits need not be independent of the validation
% outcomes. The asymptotic argument instead uses a conditional Donsker
% condition on the class indexed by the fluctuation parameters; see
% Theorem~\ref{thm:asymptotic-properties} in the Supplementary Materials.}
The targeted minimum loss-based estimator $\thetatmle$ is computed as
a substitution estimator that uses an estimate $\tilde\Q_{1,j(i)}$
constructed to solve the cross-validated efficient influence function
estimating equation
$\Pn \{\D_1(\cdot, \tilde\eta_{j(\cdot)}) - \thetatmle\}= 0$. The
construction of $\tilde\Q_{1,j(i)}$ is motivated by the observation
that, for an appropriately defined variable
$\omega(\bar s_t,a_t, h_t)$ the efficient influence function of
$\theta$ can be expressed as a sum of terms of the form
$\omega(\bar s_t,a_t, h_t)\{q_{t+1}(\bar s_t, a_{t+1}, h_{t+1}) -
\Q_t(\bar s_t, a_t, h_t)\}$, which, for fixed $\bar s_t$, takes the
form of score functions $\omega(W)\{M - \E(M\mid W)\}$ for
appropriately defined variables $M$ and $W$ and some weight function
$\omega$. %It iswell known that i
If $\E(M\mid W)$ is estimated within a weighted generalized linear
model with canonical link that includes an intercept, then the
weighted MLE estimate solves the score equation
$\sum_i\omega(W_i)\{M_i - \hat\E(M_i\mid W_i)\}=0$. TMLE uses this
observation to iteratively tilt preliminary estimates of $\Q_t$
towards a solution of the efficient influence function estimating
equation. The algorithm is defined as follows:

\begin{enumerate}[label=Step \arabic*., align=left, noitemsep, leftmargin=*]
\item Initialize $\tilde\eta =\hat\eta$ and
  $\tilde m_{\tau+1,j(i)}=\tilde q_{\tau+1,j(i)} = Y_i$.
\item For $k=1,\ldots,\tau$, and for each $\bar s_k$, compute the weights
  \[\omega_{k,i}(\bar s_k) = D_{0,k}(\bar s_k)\left(\prod_{u=1}^k
      \frac{\hat g_u(s_u\mid H_u)}{\hat g_u(A_u\mid
        H_u)}\right).\]
\item For $t=\tau,\ldots,1$, for each $\bar s_t$:
  \begin{itemize}[leftmargin=0pt]
  \item Fit the generalized linear tilting model
    \[\link \tilde\Q_t^\epsilon(\bar s_t, A_{t,i},H_{t,i}) = \epsilon +
    \link \tilde\Q_{t,j(i)}(\bar s_t, \edit{A_{t,i}}, H_{t,i}),\]
    where
    $\link(\cdot)$ is the canonical link. The parameter $\epsilon$ may
    be estimated by running a generalized linear model of the
    pseudo-outcome
    $\tilde q_{t+1,j(i)}(\bar s_t, A_{t+1,i}, H_{t+1,i})$ with only an
    intercept term, an offset term equal to
    $\link \tilde\Q_{t,j(i)}(\bar s_t, A_{t,i},H_{t,i})$, and weights
    $\omega_{t,i}(\bar s_t)$, using all the data points in the
    sample. 
  \item Let $\hat\epsilon$ denote the \edit{estimate}, and update the regression
    estimates as
    \begin{align*}
      \link \tilde\Q_{t, j(i)}^{\hat\epsilon}(\bar s_t, A_{t,i},H_{t,i})&=
                                                                \hat\epsilon
                                                                     +
                                                                     \link\tilde\Q_{t,
                                                                     j(i)}(\bar s_t,
                                                                     A_{t,i},H_{t,i})\\
      \link \tilde q_{t, j(i)}^{\hat\epsilon}(\bar s_{t-1}, A_{t,i}, H_{t,i})&=
                                                                   \hat\epsilon
                                                                          +
                                                                          \link\tilde\Q_{t,
                                                                          j(i)}((\bar s_{t-1}, A_{t,i}),d((\bar s_{t-1}, A_{t,i}),
                                                                          H_{t,i}),
                                                                          H_{t,i}).
    \end{align*}
  \item Update $\tilde\Q_{t, j(i)}= \tilde\Q_{t, j(i)}^{\hat\epsilon}$, $t = t-1$, and iterate. 
  \end{itemize}
\item The TMLE is defined as
  $\thetatmle=\frac{1}{n}\sum_{i=1}^n\tilde q_{1,j(i)}(A_{1,i},H_{1,i})$.
\end{enumerate}
In Step 3, an outcome bounded in an interval $[a,b]$ may be analyzed
with logistic regression (i.e., $\link=\logit$) by mapping it to an
outcome $(0,1)$ through the transformation
$(Y - a) / (b - a)(1-2\epsilon) + \epsilon$ for some small value
$\epsilon>0$. This approach has robustness advantages compared to
fitting a linear model \citep[see ][]{Gruber2010t}. \edit{The
  iterative procedure and the score-equation argument above guarantee
  that the cross-fitted efficient influence function estimating
  equation is solved. Theorem~\ref{thm:asymptotic-properties} in the
  Supplementary Materials gives conditions under which \(\thetatmle\)
  is consistent and satisfies
  $ \sqrt{n}(\thetatmle-\theta) = \frac{1}{\sqrt{n}} \sum_{i=1}^n
  \overline\varphi_1(Z_i;\eta) + o_P(1)$.  Consequently,
  \(\thetatmle\) is asymptotically normal and efficient. The proof
  adapts standard arguments for cross-fitted TMLE
  \citep[e.g.,][]{zheng2011cross,diaz2023nonparametric,kennedy2024semiparametric}
  to the proposed estimator.  }

%The iterative procedure and the score equation argument above
%guarantee that the above estimates solve the efficient influence
%function estimating equation, which is crucial to prove that
%$\sqrt{n}(\thetatmle-\theta)$ converges to a random variable
%$N(0;\var[\bar\D_1(Z;\eta)])$, provided that
%$\rem_0(\hat\eta)=o_P(n^{-1/2})$ and that the true and estimated
%weights $\omega_{t}(s_i)$ are bounded. { A formal
%  statement of sufficient conditions for consistency, asymptotic
%  normality, and efficiency, together with a proof, is provided in
%  Section~\ref{sec:asymptotic-properties} of the Supplementary
%  Materials. The proof adapts standard arguments for cross-fitted TMLE
%  \citep[e.g.,][]{zheng2011cross,diaz2023nonparametric} to the
%  proposed history-dependent LMTP estimator.  }

\subsection{Sequential regression estimator using SDR unbiased
  transformations}\label{sec:sdr}

In this section we use the \edit{sequential doubly robust} unbiased
transformation in expression (\ref{eq:fo}) to obtain an estimate of
$\Q_t$.  We say that $\D_{t+1}$ is a \edit{sequential doubly robust} unbiased
transformation
\citep{buckley1979linear,rubin2007doubly,diaz2013targeted,kennedy2017non,luedtke2017sequential}
for $\Q_t$ due to the following proposition, which is a
straightforward consequence of Theorem \ref{theo:sdr}.
\begin{proposition}
  Let $\eta'$ be such that either $\Q'_\ell=\Q_\ell$ or $\r'_\ell=\r_\ell$ for all
  $\ell>t$. Then we have
  $\E\big[\D_{t+1}(\bar s_t, Z;\eta')\mid A_t=a_t,H_t=h_t\big] =\Q_t(\bar s_t, a_t,h_t)$.
\end{proposition}
The above motivates the construction of the sequential regression
estimator by regressing an estimate of the data
transformation $\D_{t+1}(\bar s_t, Z;\eta)$ on $(A_t,H_t)$ for each $\bar s_t$, starting at
$\D_{\tau+1}(Z;\eta)=Y$. For
preliminary cross-fitted estimates
$\hat\r_{1,j(i)},\ldots,\hat \r_{\tau, j(i)}$, the estimator is
defined as follows:
\begin{enumerate}[label=Step \arabic*, align=left, leftmargin=*]
\item Initialize  $\D_{\tau+1}(\bar s_\tau,
  Z_i;\underline{\check\eta}_{\tau,j(i)})= Y_i$ for $i=1,\ldots,n$.
\item \edit{For $j = 1,\ldots,J$ and $t = \tau,\ldots,1$, for each
    $\bar s_t$: compute the pseudo-outcome
    $\check Y_{t+1,i} = \varphi_{t+1}(\bar s_t, Z_i;
    \check\eta_{t,j})$ for $i \in \mathcal{T}_j$, regress
    $\check Y_{t+1,i}$ on $(A_{t,i}, H_{t,i})$ using only
    $i \in \mathcal{T}_j$, and let $\check m_{t,j}$ denote the
    output. All nuisance estimates for fold $j$ depend only on
    the training sample $\mathcal{T}_j$.\label{step:2}}
\item Define the sequential estimator regression as
  $\thetasr=\frac{1}{n}\sum_{i=1}^n\D_1(Z_i;\check\eta_{j(i)})$.
\end{enumerate}

\edit{At the terminal step, a binary outcome is fit using a binomial
  regression. Intermediate pseudo-outcomes are treated as continuous
  and fit using regression learners with an identity link. For
  terminal outcomes $Y$ bounded in $[0,1]$, the resulting regression
  predictions are truncated to $(\epsilon,1-\epsilon)$, for a small
  numerical constant $\epsilon>0$, before being used in subsequent
  steps of the
  recursion. %For a fixed clipping constant, the nuisance-convergence
  %conditions require the true regression means to lie within the
  %clipping interval; otherwise a vanishing clipping sequence is
  %needed.
  Note that truncating the nuisance regression predictions in the SDR
  does not constrain the SDR estimate itself to the parameter space,
  since the estimator additionally contains augmentation terms.  In
  contrast, the TMLE is a substitution estimator and, with the
  logistic fluctuation used here, respects the bounds of the target
  parameter by construction.}
\edit{Theorem~\ref{thm:asymptotic-properties} in the Supplementary
  Materials also gives conditions under which \(\thetasr\) is
  consistent and has the same efficient asymptotic distribution as
  \(\thetatmle\). Under these conditions, an asymptotically valid
  \((1-\alpha)100\%\) confidence interval is
  $ \thetasr \ \pm\ z_{1-\alpha/2}\frac{\widehat\sigma}{\sqrt{n}}$,
  where \(\widehat\sigma^2\) is the empirical variance of
  \(\overline\varphi_1\{Z_i;\check\eta_{j(i)}\}\).}

%Like its TMLE counterpart, this estimator satisfies $\sqrt{n}(\thetasr-\theta)$ converges to a random variable
%$N(0;\var[\bar\D_1(Z;\eta)])$, provided that
%$\rem_0(\hat\eta)=o_P(n^{-1/2})$. This can be used to construct
%asymptotically valid confidence intervals and hypothesis tests, for example a $(1-\alpha)100\%$ confidence interval may be constructed as $\thetasr\pm z_{1-\alpha/2}\hat\sigma/\sqrt{n}$, where $\hat\sigma^2$ is the empirical variance of $\bar\D_1(Z_i;\check\eta_{j(i)})$.
%%%%%%%%%%%%%%%%%%%%%%%%%%%%%%%%%%%%%%%%%%%%%%%%%%%%%%%%%%%%%%%%%%%%%%%%%%%%%%% 

\section{Numerical illustration}\label{sec:illustrate}

We use Monte Carlo simulation to illustrate the performance of the
TMLE and SDR estimators \edit{in situations where such performance is
  predicted by theory}. For
$n \in \left\{250, 500,1000,5000,10000 \right\}$, data were simulated
from the data generating mechanism given in Figure~\ref{fig:dgp}.
%\begin{align*}
%  L_t &= (0.5 L_{t-1} + \epsilon_t)\one\{Y_{t-1} = 0\}
%        +\text{NULL}\one\{Y_{t-1} = 1\}\\
%  P(A_t = 1 \mid L_t) &= \text{logit}^{-1}(-1.5 + 0.3 L_t)
%                        \one\{A_{t-1} = 0, Y_{t-1} = 0\} \one\{A_{t-1}
%                        = 1,  Y_{t-1}=0\} + A_{t-1}\one\{Y_{t-1} = 1\}
%  \\
%  P(Y_t = 1 \mid A_t, L_t) &= \text{logit}^{-1}(-2 + 0.4 L_t - 0.8
%                             A_t) \one\{ Y_{t-1} = 0 \}+ \one\{Y_{t-1}
%                             = 1\}.%
                              %\end{align*}
\begin{figure}[!htp]
Let \edit{$L_0 \sim \text{Bernoulli}(0.5)$}
and $A_0 = 0$, $Y_0 = 0$. For $t \in \{1, \dots, 5\}$, define:
\begin{align*}
\edit{P(L_t = 1 \mid L_{t-1},Y_{t-1})} &= \edit{\begin{cases} 0.5 L_{t-1} & \text{if } Y_{t-1} = 0 \\ \text{undefined} & \text{if } Y_{t-1} = 1 \end{cases}}\\
\edit{P(A_t = 1 \mid L_t,A_{t-1},Y_{t-1})} &= \begin{cases} \text{logit}^{-1}(-1.5 + 0.3 L_t) & \text{if } A_{t-1} = 0 \text{ and } Y_{t-1} = 0 \\ 1 & \text{if } A_{t-1} = 1 \text{ and } Y_{t-1}=0\\ A_{t-1} & \text{if } Y_{t-1} = 1\end{cases} \\
\edit{P(Y_t = 1 \mid A_t,L_t,Y_{t-1})} &= \begin{cases} \text{logit}^{-1}(-2 + 0.4 L_t - 0.8 A_t) & \text{if } Y_{t-1} = 0 \\ 1 & \text{if } Y_{t-1} = 1. \end{cases}
\end{align*}
\caption{Data generating \edit{mechanism} for the numerical illustration in
  \S\ref{sec:illustrate}}
\label{fig:dgp}
\end{figure}
This data generating mechanism corresponds to a survival analysis
scenario with time-varying confounding and conditionally degenerate
exposure assignment. At each sample size, we generated $1000$
data-sets, and for each data-set, using both the SDR and TMLE, we
estimated the expected value of $Y_5$ under a hypothetical
intervention where exposure ($A_t$) was delayed by one
time-point. Outcome regressions were estimated using a \edit{Super
  Learner ensemble of a generalized linear model with all possible
  two-way interactions and cross-validated LASSO with main terms
  \citep{vanderLaan&Polley&Hubbard07,tibshirani1996regression}}; the
propensity score at each time-point was estimated using a main-term
generalized linear model with the logit link. Because the robustness
results of the SDR and TMLE have been thoroughly discussed in the
existing literature, (see e.g., \cite{diaz2023nonparametric,
  wen2023intervention, diaz2024survival}) we limit the simulation
study to the ideal scenario where all nuisance parameters are
correctly specified and estimated at fast enough rates. \edit{Because
  $L_t$ is binary and the data-generating mechanism is first-order
  Markov, for each fixed natural-treatment history $\bar s_t$ the
  sequential regressions depend on the past only through $(A_t, L_t)$,
  so the Super Learner library contains a saturated, and hence
  correctly specified, candidate.} Performance was summarized using
absolute bias \edit{$\times 100$, mean squared error (MSE)
  $\times 100$}, $n \times \text{MSE}$, and nominal 95\% confidence
interval coverage. The true value under the intervention was
calculated using Proposition \ref{propo:sr} with correctly specified
models and a ``super-population'' of one-million observations. The
\texttt{R} code for the simulation can be found at
\href{https://github.com/CI-NYC/htlmtp-simulations}{https://github.com/CI-NYC/htlmtp-simulations}.

\edit{The results are presented in Table~\ref{tab:sim}. In this
scenario, the TMLE and SDR performed similarly across the sample sizes
considered. Their mean squared errors decreased with increasing sample
size, while $n\times\mathrm{MSE}$ remained approximately constant.
Both estimators had confidence interval coverage close to nominal.
We caution against drawing general conclusions about comparative
finite-sample performance from this single simulation setting, since
performance can vary with the data-generating mechanism
\citep{rudolph2023all}.}

\begin{table}[H]
    \centering
    \caption{Performance of the SDR and TMLE estimators in a simulation study.}\label{tab:sim}
    \edit{\begin{tabular}[t]{llcccc}
        \toprule
        $n$ & Estimator & $\left|\text{Bias}\right| \times 100$ & $\text{MSE} \times 100$ & $n \times \text{MSE}$ & 95\% Cov.\\        
        \midrule
        \multirow{2}{*}{250} & SDR & 0.048 & 0.196 & 0.490 & 0.953\\
         & TMLE & 0.065 & 0.197 & 0.491 & 0.953\\     
         \addlinespace[0.3em]  
         \multirow{2}{*}{500} & SDR & 0.016 & 0.097 & 0.485 & 0.951\\
         & TMLE & 0.013 & 0.097 & 0.485 & 0.949\\
        \addlinespace[0.3em]
        \multirow{2}{*}{1,000} & SDR & 0.016 & 0.049 & 0.494 & 0.954\\
         & TMLE & 0.016 & 0.049 & 0.495 & 0.952\\
        \addlinespace[0.3em]
        \multirow{2}{*}{5,000} & SDR & 0.022 & 0.009 & 0.473 & 0.953\\
         & TMLE & 0.021 & 0.009 & 0.473 & 0.953\\
        \addlinespace[0.3em]
        \multirow{2}{*}{10,000} & SDR & 0.006 & 0.005 & 0.500 & 0.939\\
         & TMLE & 0.006 & 0.005 & 0.501 & 0.939\\
        \bottomrule
    \end{tabular}}
\end{table}

\section{Motivating application}
%%%%%%%%%%%%%%%%%%%%%%%%%%%%%%%%%%%%%%%%%%%%%%%%%%%%%%%%%%%%%%%%%%%%%%%%%%%%%%% 
We applied our proposed approach \edit{to higher-risk pain treatment
to estimate the effect of withholding the first naturally indicated treatment month on the
12-month probability of remaining free from opioid use disorder (OUD)
or overdose} among \edit{non-dual-eligible},
Medicaid-insured patients, aged 19-63 years, newly treated for pain due to lumbar spinal
stenosis, N=12745. Day 0 was the first day of the pain treatment claim,
which occurred within 30 days of a new-onset diagnosis of lumbar
spinal stenosis (ICD-10 code M4806). Patients must have been
continuously enrolled in Medicaid for the 6 months prior to day 0 (a
``washout period") to determine cohort eligibility, which, in addition
to the above age and eligibility criteria, also included being
opioid-naive, with no history of OUD, no cancer diagnosis, not
pregnant, and community-dwelling. Additionally, they could not have
been hospitalized in the 1 month prior to day 0. This cohort was drawn
from Medicaid T-MSIS Analytic Files (TAF), Demographics, Other
Services, Inpatient, and Pharmacy claims, 2016-2019.
\edit{We considered the following 26 states that had implemented
Affordable Care Act (ACA) Medicaid expansion by the end of 2014:
AR, AZ, CA, CO, CT, DE, HI, IA, IL, KY, MA, MD, MI, MN, ND, NH, NJ,
NM, NV, NY, OH, OR, RI, VT, WA, WV. Maryland was excluded due to poor
diagnostic code quality, leaving 25 states in the analysis.}

We were interested in examining \edit{this treatment-withholding policy} among
patients with a lumbar spinal stenosis diagnosis, because these
patients typically receive multiple months of pain treatment and
frequently receive pain treatments with higher risks of opioid-related harms. % where concerns have been raised about their OUD risks.}
 Starting from day 0, we discretized follow-up
into 1-month intervals, ending at month 12. During each month, we
evaluated 1) which pain treatments the patient had a claim for, 2)
whether or not the patient was censored, and 3) presence of an OUD or
overdose diagnosis code or receipt of medication for OUD.
\edit{These codes and medication claims defined the composite OUD or
overdose endpoint. For this analysis, $Y=1$ indicates remaining free
from that endpoint through month 12.} We
collected information on the following baseline covariates during the
washout period: age; sex; race/ethnicity; receipt of: \edit{Temporary
  Assistance for Needy Families (TANF)} benefits, \edit{Supplemental
  Security Income (SSI)} benefits; veteran status; household size;
marital status; \edit{English} as primary language; diagnosis codes for:
anxiety, bipolar, depression, ADHD, other psychiatric illness, alcohol
use disorder, or other substance use disorder; mental health
counseling claims; healthcare utilization variables: number of
inpatient hospitalizations in the washout period prior to one month
before day 0, number of outpatient visits, number of emergency
department visits.  We considered the following set of pain treatment
variables: i) opioid analgesic prescription of $\le7$ days and dose
$\le50$ morphine milligram equivalents (MME), ii) opioid analgesic
prescription of $>7$ days and dose $\le50$MME, iii) opioid analgesic
prescription of dose $>50$MME, iv) \edit{gabapentinoid} prescription, v)
benzodiazepine prescription, vi) muscle relaxant prescription, vii)
duloxetine prescription, viii) anti-inflammatory prescription, ix)
interventional procedure, x) chiropractic visit, xi) physical therapy
visit. Patients could have, and frequently did have, more than one
treatment during a given month. For the purpose of this illustrative
analysis we dichotomized treatment into the presence vs absence of any
of the following higher-risk treatments, where this binary treatment
variable, $A_t=1$ if any of treatments ii-v were present in month $t$
and $A_t=0$ otherwise.  Patients were censored in month $t$ if: they
had disenrolled from Medicaid by the end of month $t$, they became
dual-eligible for Medicare, or at the end of 2019. The study was approved by
the Columbia University Institutional Review Board. All code to
replicate cohort and variable creation and statistical analyses is
given in the GitHub repository:
\url{https://github.com/CI-NYC/m4806-diagnosis}.

\edit{We compared the policy that withholds the first naturally
  indicated treatment month with the natural treatment process, under
  elimination of censoring in both cases. The contrast was
  $\E\{Y(\bar A^d,\bar C=\bar 1)-Y(\bar C=\bar 1)\}$, where
  $ d_t(\bar s_t,h_t)=0$ if
  $s_t=1\text{ and }s_\ell=0\text{ for all }\ell<t$ and
  $d_t(\bar s_t,h_t) = s_t$ otherwise. The policy leaves subsequent
  natural treatment decisions unchanged; it does not require treatment
  to begin in the following month.} \edit{Right censoring was handled
  as described in Supplementary Section~\ref{sec:right-censoring}:
  \(C_t\) was treated as an additional intervention variable and set
  to one in every month, and the conditional probability of remaining
  uncensored was estimated as an additional nuisance function.}
% \[
%   d_t(a_t, \bar a_{t-1}) =
%   \begin{cases}
%     0 & \text{if } a_t = 1 \text{ and } a_\ell =0 \text{ for all } \ell < t, \\
%     a_t & \text{otherwise}.
%   \end{cases}
% \]
We estimated the effect of \edit{this withholding policy} using both
the TMLE and SDR estimators with 2-fold cross-fitting. The nuisance
parameters $\r_t$ and $\Q_t$ were estimated using the Super Learner
\citep{vanderLaan&Polley&Hubbard07} with a GLM and multivariate
adaptive regression splines (MARS) as candidate algorithms. \edit{GLM
  learners used logit links for treatment, censoring, and terminal
  binary outcome regressions, and identity links for intermediate
  outcome regressions. For TMLE, these initial outcome fits were
  subsequently updated using a logistic fluctuation.} Results are
shown in Table~\ref{tab:medicaid}. Using both the TMLE and SDR to
adjust for right-censoring, we estimated the 12-month probability of
remaining free from an OUD or overdose diagnosis (\edit{survival-type}
probability) as 0.963 (95\% CI: 0.960, 0.966). \edit{Under the policy
  withholding the first naturally indicated treatment month}, the TMLE
estimated the 12-month probability of remaining free from \edit{the
  composite OUD or overdose endpoint} as \edit{0.986 (95\% CI: 0.915,
  1)}, where the upper bound was truncated at one. The SDR produced a
point estimate of 1.27 \edit{(95\% CI: 1.19, 1.35)}, outside the
parameter space. This is a well-known limitation of the SDR when the
true value lies near the boundary due to it not being a substitution
estimator (\S\ref{sec:estima}). \edit{ Although the consistency
  assumptions required for SDR are weaker than those for TMLE,
  \edit{TMLE can sometimes be preferred}. For bounded parameters,
  particularly near a boundary as in this illustrative application,
  TMLE has the practical advantage of producing an interpretable
  estimate within the parameter space.}

\begin{table}[ht]
\caption{\edit{Estimated 12-month OUD- or overdose-free probability under elimination of censoring.}}\label{tab:medicaid}
\centering
\small
\edit{\begin{tabular}{l cc cc ccc}
\toprule
 & \multicolumn{2}{c}{Reference} & \multicolumn{2}{c}{\shortstack{First indicated\\month withheld}} & \multicolumn{3}{c}{Withholding effect} \\
\cmidrule(lr){2-3} \cmidrule(lr){4-5} \cmidrule(lr){6-8}
Estimator & Estimate & 95\% CI & Estimate & 95\% CI & Estimate & 95\% CI & P \\
\midrule
TMLE & 0.963 & 0.960--0.966 & 0.986 & 0.915--1 & 0.022 & -0.048--0.092 & 0.53 \\
SDR  & 0.963 & 0.960--0.966 & 1.27 & 1.19--1.35 & 0.31 & 0.22--0.39 & $<$0.001 \\
\bottomrule
\end{tabular}}
\end{table}

\section{Conclusion}\label{sec:discussion}

% This paper proposes targeted learning estimators for longitudinal
% modified treatment policies whose hypothetical interventions may
% depend on the \edit{sequence of past natural treatment values}.  This
% extension allows the double/de-biased machine learning estimation of
% the effects of policies that modify the timing or trajectory of
% treatment decisions. The proposed augmented sequential regression
% representation, together with the corresponding efficient influence
% function, yields practical TMLE and SDR estimators for this broader
% class of interventions.
% We illustrated the sampling distribution of
%the estimators in a numerical study, and demonstrated its use in an
%application to examining the effect of delaying risky pain treatments
%on opioid use disorder risk.

Several directions of research remain open. We have focused on
discrete exposures, for which the relevant natural treatment histories
can be represented by summing over possible exposure sequences. Many
applications motivating modified treatment policies involve
continuous, multivariate, or mixed discrete-continuous treatments,
including medication doses and combinations of therapies. Future work
should develop efficient estimators for G-LMTPs in these settings,
avoiding explicit enumeration of treatment histories and replacing
conditional mass-function ratios with appropriate density-ratio or
Riesz representers. Other useful directions include scalable
algorithms for long follow-up, \edit{the development of a sequentially
  doubly robust TMLE similar to \cite{luedtke2017sequential}},
diagnostics for support violations induced by the proposed policies,
and extensions to settings with competing risks, interference, or
continuous time.

\section*{Acknowledgments}
The authors are grateful to Alex Luedtke for insightful conversations
that helped clarify the motivation for the generalization developed in
this paper, as well as Anton Hung for the creation of the dataset used
in our motivating application.

Iván Díaz was supported by a Patient-Centered Outcomes Research
Institute (PCORI) Project Program Award (ME-2024C3-43044). Pawe{\l}
Morzywo{\l}ek was partially supported through a National Institutes of
Health (NIH) Award (5UM1AI068635-20), the Pacific Institute for the
Mathematical Sciences (PIMS), UW eScience Institute and the Pioneer Center for Statistical and computational Methods for Advanced Research to Transform Biomedicine (SMARTbiomed), DNRF grant number P4.  Kara Rudolph was supported by a National
Institute on Drug Abuse (NIDA) Award (R01DA053243).

\section*{Disclosure of AI use}
The authors used ChatGPT to help with writing and checking of
mathematical details. \edit{The authors} reviewed and revised the
AI-assisted material and have attempted to make the paper accessible
and well-written. They are responsible for all of the manuscript
contents.

\backmatter
\bibliographystyle{plainnat}
\bibliography{refs}
\label{lastpage}

\clearpage
\section*{Supplementary Materials for\\ \titlepaper}
\addcontentsline{toc}{section}{Supplementary Materials}
\label{sm:firstpage}
% Preserve the separate supplement's printed numbering while using
% unique anchors for cross-references within the combined document.
\counterwithout*{example}{section}
\setcounter{section}{0}
\renewcommand{\theHsection}{supp.\arabic{section}}
\setcounter{subsection}{0}
\renewcommand{\theHsubsection}{supp.\arabic{subsection}}
\setcounter{equation}{0}
\renewcommand{\theHequation}{supp.\arabic{equation}}
\renewcommand{\theequation}{SM\arabic{equation}}
\setcounter{figure}{0}
\renewcommand{\theHfigure}{supp.\arabic{figure}}
\renewcommand{\thefigure}{SM\arabic{figure}}
\setcounter{table}{0}
\renewcommand{\theHtable}{supp.\arabic{table}}
\renewcommand{\thetable}{SM\arabic{table}}
\setcounter{theorem}{0}
\renewcommand{\theHtheorem}{supp.\arabic{theorem}}
\renewcommand{\thetheorem}{SM\arabic{theorem}}
\setcounter{proposition}{0}
\renewcommand{\theHproposition}{supp.\arabic{proposition}}
\renewcommand{\theproposition}{SM\arabic{proposition}}
\setcounter{lemma}{0}
\renewcommand{\theHlemma}{supp.\arabic{lemma}}
\renewcommand{\thelemma}{SM\arabic{lemma}}
\setcounter{coro}{0}
\renewcommand{\theHcoro}{supp.\arabic{coro}}
\renewcommand{\thecoro}{SM\arabic{coro}}
\setcounter{assumption}{0}
\renewcommand{\theHassumption}{supp.\arabic{assumption}}
\renewcommand{\theassumption}{SM\arabic{assumption}}
\setcounter{assumptioniden}{0}
\renewcommand{\theHassumptioniden}{supp.\arabic{assumptioniden}}
\setcounter{definition}{0}
\renewcommand{\theHdefinition}{supp.\arabic{definition}}
\renewcommand{\thedefinition}{SM\arabic{definition}}
\setcounter{remark}{0}
\renewcommand{\theHremark}{supp.\arabic{remark}}
\renewcommand{\theremark}{SM\arabic{remark}}
\setcounter{example}{0}
\renewcommand{\theHexample}{supp.\arabic{example}}
\renewcommand{\theexample}{SM\arabic{example}}

\section{Identification (Theorem \ref{theo:iden})}
\edit{Throughout this supplement, $\bar s_t=(s_1,\ldots,s_t)$ denotes
retained natural treatment inputs, $a_t$ denotes a treatment coordinate
in a regression or observed-data history, and $a_t^d$ denotes the
policy-assigned treatment. These arguments are kept distinct, including
when the definition of $q_t$ substitutes $s_t=a_t$ before modifying the
treatment argument of $m_t$.}
\begin{proof}
\edit{
  Recall that for fixed natural treatment values $\bar s_t$ and
  covariate values $\bar l_t$, we denote
  $a_t^\d=\d_t(\bar s_t, h^\d_t)$, where
  $h^\d_t=(\bar a_{t-1}^\d, \bar l_t)$.  Define the random variable
  \[
  \begin{aligned}
    Z_t(\bar s_t)=Y\bigl\{\bar a_t^\d,\,
    \underline d_{t+1}\bigl(&\bar s_t,\,
      \underline A_{t+1}(\bar a_t^\d,\underline A_{t+1}^\d),\\
      &H_{t+1}(\bar a_t^\d,\underline A_{t+1}^\d)\bigr)\bigr\}.
  \end{aligned}
  \]
 \edit{$Z_t (\bar s_t)$ is the counterfactual outcome obtained by
   specifying the first \(t\) natural treatment values as
   \(\bar s_t\), applying the policy to them, and then continuing the
   policy from time \(t+1\) onward.} Note that
 $Z_0(s_0)=Y(\bar A^\d)$, and that
 $Z_t(\bar s_t)=Z_{t+1}(\bar s_{t+1})$ in the event $A_{t+1}=s_{t+1}$
 and $H_{t+1} = h_{t+1}^\d$, and that $Z_\tau(\bar s_\tau) = Y$ on the
 event $A_\tau=a_\tau^\d$ and $H_\tau = h_\tau^\d$. \edit{These
   identities follow from the recursive construction of the
   counterfactual outcomes and consistency. Recall that $\bar s_t$
   records the first $t$ natural treatment values, and $\bar a_t^\d$
   denotes the corresponding treatment values assigned by the policy.

First, at $t=0$, the treatment prefix is empty, so the policy
generates the entire treatment trajectory. Thus,
$Z_0(s_0)=Y(\bar A^\d)$, where $s_0=\emptyset$.

Second, consider $t<\tau$ and the event $A_{t+1}=s_{t+1}$ and
$H_{t+1}=h_{t+1}^\d$. On this event, the observed history agrees
with the history generated by assigning the first $t$ treatments
as $\bar a_t^\d$. Consistency therefore gives
$H_{t+1}(\bar a_t^\d)=H_{t+1}=h_{t+1}^\d$ and
$A_{t+1}(\bar a_t^\d)=A_{t+1}=s_{t+1}$.
Consequently, in constructing $Z_t(\bar s_t)$, the next natural
treatment value is $s_{t+1}$, and the policy assigns
$d_{t+1}((\bar s_t,s_{t+1}),h_{t+1}^\d)=a_{t+1}^\d$.
This is exactly the assignment used in constructing
$Z_{t+1}(\bar s_{t+1})$. The two constructions therefore agree
through time $t+1$ and continue with the same policy and exogenous
variables. Their subsequent trajectories and outcomes coincide,
giving $Z_t(\bar s_t)=Z_{t+1}(\bar s_{t+1})$ on the stated event.

Finally, at $t=\tau$, no further treatment assignments remain, so
$Z_\tau(\bar s_\tau)=Y(\bar a_\tau^\d)$. On the event
$A_\tau=a_\tau^\d$ and $H_\tau=h_\tau^\d$, the entire observed
treatment and covariate history agrees with the corresponding
intervened history. Consistency therefore implies
$Z_\tau(\bar s_\tau)=Y(\bar a_\tau^\d)=Y$.} We have
  \begin{align}
    E[Y(\bar A^\d)] & =\int_{\mathcal A_1, \mathcal L_1} \E[Z_0(s_0)\mid A_1=s_1, L_1=l_1]\dd
               \P(s_1,l_1)\notag\\
             & =\int_{\mathcal A_1, \mathcal L_1} \E[Z_1(s_1)\mid A_1=s_1, L_1=l_1]\dd
               \P(s_1,l_1)\notag\\
                    & =\int_{\mathcal A_1, \mathcal L_1} \E[Z_1(s_1)\mid A_1=a_1^\d, L_1=l_1]\dd
                      \P(s_1,l_1)\notag\\
             & =\int_{\mathcal {\bar A}_2, \mathcal {\bar L}_2}
               \E[Z_1(s_1)\mid A_2=s_2,
               L_2=l_2,A_1=a_1^\d,L_1=l_1]\dd \P(s_2,l_2\mid a_1^\d,
               l_1)\dd \P(s_1,l_1)\notag\\
                    & =\int_{\mathcal {\bar A}_2, \mathcal {\bar L}_2}
                      \E[Z_2(\bar s_2)\mid A_2=a_2^\d,
                      L_2=l_2,A_1=a_1^\d,L_1=l_1]\dd \P(s_2,l_2\mid a_1^\d,
                      l_1)\dd \P(s_1,l_1)\notag\\
                    & =\int_{\mathcal {\bar A}_2, \mathcal {\bar L}_2}
                      \E[Z_2(\bar s_2)\mid A_2=a_2^\d,
                      H_2=h_2^\d]\dd \P(s_2,l_2\mid a_1^\d,
                      l_1)\dd \P(s_1,l_1)\notag\\
             &\,\,\, \vdots\notag\\
             & =\int_{\mathcal {\bar A}_\tau, \mathcal {\bar L}_\tau}
               \E[Z_\tau(\bar s_{\tau})\mid A_\tau=a_\tau^\d,
               H_\tau=h_\tau^\d] \prod_{k=1}^\tau\dd\P(s_k,l_k\mid
               a_{k-1}^\d, h_{k-1}^\d)\notag\\
             & =\int_{\mathcal {\bar A}_\tau, \mathcal {\bar L}_\tau} \E[Y\mid A_\tau=a_\tau^\d,
               H_\tau=h_\tau^\d] \prod_{k=1}^\tau\dd\P(s_k,l_k\mid
               a_{k-1}^\d, h_{k-1}^\d)\label{eq:eq1}
  \end{align}
\edit{The first and fourth equalities follow from the law of iterated
expectation, and the second follows from consistency.
Assumption~\ref{ass:exch}, through the conditional independence
in Lemma~\ref{lemma:indep}, justifies the third equality:
\[
E\{Z_1(s_1)\mid A_1=s_1,L_1=l_1\}
=
E\{Z_1(s_1)\mid A_1=a_1^d,L_1=l_1\}.
\]
The fifth equality combines consistency with the same
conditional-independence argument at time \(2\), and these arguments
are repeated at subsequent times.  Assumption~\ref{ass:support}
ensures that the modified treatment--history pairs \((a_t^d,h_t^d)\)
encountered in this recursion belong to the observed-data support, so
that the required observed-data conditional distributions can be
evaluated at those pairs. For $k=1$, the corresponding factor in the product is the marginal
distribution $\dd\P(s_1,l_1)$. The final equality follows from
consistency.}
}
\end{proof}

\begin{lemma}\label{lemma:indep}
  \edit{Under Assumption~\ref{ass:exch}, we have $Z_t(\bar s_t)\indep
  A_t\mid H_t$.}
\end{lemma}
\begin{proof}
\edit{
  This lemma follows trivially after noticing that under the assumed
  NPSEM, $Z_t(\bar s_t)$ is a deterministic function
  of $(\underline U_{L,t+1}, \underline U_{A,t+1}, U_Y)$.
}
\end{proof}

\subsection*{Illustration of identification}

\edit{Consider $\tau=3$ timepoints with no baseline covariates, so
$H_1$ is empty. For fixed natural treatment inputs $\bar s_3$ and
covariate values $(l_2,l_3)$, write
\[
a_1^d=d_1(s_1,H_1),\qquad
a_2^d=d_2(\bar s_2,h_2^d),\qquad
a_3^d=d_3(\bar s_3,h_3^d),
\]
where $h_2^d=(a_1^d,l_2)$ and
$h_3^d=(a_1^d,a_2^d,l_2,l_3)$. Define
$m_3(a_3,h_3)=\E[Y\mid A_3=a_3,H_3=h_3]$. The identifying
functional in Theorem~\ref{theo:iden} is
\begin{align*}
\theta
={}&\int m_3(a_3^d,h_3^d)\,
\dd\P(s_3,l_3\mid A_2=a_2^d,H_2=h_2^d)\\
&\qquad\times\dd\P(s_2,l_2\mid A_1=a_1^d)\,\dd\P(s_1).
\end{align*}
Thus, integration is over the natural inputs $s_1,s_2,s_3$, while
the conditioning histories contain the assigned treatments.

The sequential regression representation starts with
\[
q_3(\bar s_2,a_3,h_3)
=m_3\{d_3((\bar s_2,a_3),h_3),h_3\}.
\]
Next, for each retained prefix $\bar s_2$, define
\begin{align*}
m_2(\bar s_2,a_2,h_2)
&=\E\{q_3(\bar s_2,A_3,H_3)\mid A_2=a_2,H_2=h_2\}\\
&=\int m_3\{d_3((\bar s_2,s_3),h_3),h_3\}\,
\dd\P(s_3,l_3\mid A_2=a_2,H_2=h_2),
\end{align*}
where, if $h_2=(a_1,l_2)$, then
$h_3=(a_1,a_2,l_2,l_3)$ in the integrand. The next substitution is
\[
q_2(s_1,a_2,h_2)
=m_2\{(s_1,a_2),d_2((s_1,a_2),h_2),h_2\}.
\]
Finally, define
\[
m_1(s_1,a_1,H_1)
=\E\{q_2(s_1,A_2,H_2)\mid A_1=a_1,H_1\},
\]
and
\[
q_1(a_1,H_1)=m_1\{a_1,d_1(a_1,H_1),H_1\}.
\]
Then $\theta=\E\{q_1(A_1,H_1)\}$. Here $a_t$ remains the
regression's treatment coordinate; when constructing $q_t$, that value
is retained as the natural input $s_t$ and the treatment coordinate is
replaced by the policy assignment. Figure~\ref{fig:pooled} illustrates
these substitutions using augmented data.}

{
\subsection*{Proof of Proposition 1}

\begin{proof}
For fixed \((\overline s_t,\overline l_t)\), recall the recursive definitions
\[
a_k^d=d_k(\overline s_k,h_k^d),
\qquad
h_k^d=(\overline a_{k-1}^d,\overline l_k),
\qquad
k=1,\ldots,t.
\]
We use the convention that \(\overline s_0\) and \(\overline a_0^d\)
are empty, so that \(h_1^d=l_1\). For each \(t\in\{1,\ldots,\tau\}\),
define
\[
I_t
=
\int
q_t(\overline s_{t-1},s_t,h_t^d)
\prod_{k=1}^{t}
dP\!\left(
s_k,l_k
\mid
A_{k-1}=a_{k-1}^d,
H_{k-1}=h_{k-1}^d
\right),
\]
where the factor corresponding to \(k=1\) is interpreted as the marginal distribution
\(dP(s_1,l_1)\). When \(t=1\), the first argument of \(q_1\) is empty, and hence
\[
I_1
=
\int q_1(s_1,h_1^d)\,dP(s_1,l_1)
=
\E\{q_1(A_1,H_1)\}.
\] We next show that $I_t=I_{t+1}$ for 
$t=1,\ldots,\tau-1$.  Fix \(t<\tau\). By the definition of \(q_t\),
\begin{align*}
q_t(\overline s_{t-1},s_t,h_t^d)
&=
m_t\!\left[
(\overline s_{t-1},s_t),
d_t\{(\overline s_{t-1},s_t),h_t^d\},
h_t^d
\right]
\\
&=
m_t(\overline s_t,a_t^d,h_t^d).
\end{align*}
By the recursive definition of \(m_t\),
\begin{align*}
m_t(\overline s_t,a_t^d,h_t^d)
&=
\E\!\left\{
q_{t+1}(\overline s_t,A_{t+1},H_{t+1})
\mid
A_t=a_t^d,H_t=h_t^d
\right\}
\\
&=
\int
q_{t+1}(\overline s_t,s_{t+1},h_{t+1}^d)
dP\!\left(
s_{t+1},l_{t+1}
\mid
A_t=a_t^d,H_t=h_t^d
\right).
\end{align*}
The second equality follows because, conditional on
\(A_t=a_t^d\), \(H_t=h_t^d\), and
\((A_{t+1},L_{t+1})=(s_{t+1},l_{t+1})\), the history at time \(t+1\)
is $
H_{t+1}
=
(\overline a_t^d,\overline l_{t+1})
=
h_{t+1}^d$. Substituting this expression into the definition of \(I_t\) gives
\begin{align*}
I_t
&=
\int
\left[
\int
q_{t+1}(\overline s_t,s_{t+1},h_{t+1}^d)
\,
dP\!\left(
s_{t+1},l_{t+1}
\mid
A_t=a_t^d,H_t=h_t^d
\right)
\right]
\\
&\hspace{3em}\times
\prod_{k=1}^{t}
dP\!\left(
s_k,l_k
\mid
A_{k-1}=a_{k-1}^d,
H_{k-1}=h_{k-1}^d
\right).
\end{align*}
Combining the inner integration over
\((s_{t+1},l_{t+1})\) with the outer integrations over
\((\overline s_t,\overline l_t)\) yields
\[I_t
=
\int
q_{t+1}(\overline s_t,s_{t+1},h_{t+1}^d)
\prod_{k=1}^{t+1}
dP\!\left(
s_k,l_k
\mid
A_{k-1}=a_{k-1}^d,
H_{k-1}=h_{k-1}^d
\right)
=
  I_{t+1}.\]
Therefore, $
\E\{q_1(A_1,H_1)\}
=
I_1
=
I_2
=
\cdots
=
I_\tau$.  It remains to evaluate \(I_\tau\). By the definition of \(q_\tau\),
\begin{align*}
q_\tau(\overline s_{\tau-1},s_\tau,h_\tau^d)
&=
m_\tau\!\left[
d_\tau\{(\overline s_{\tau-1},s_\tau),h_\tau^d\},
h_\tau^d
\right]
\\
&=
m_\tau(a_\tau^d,h_\tau^d).
\end{align*}
Consequently,
\[
I_\tau
=
\int
m_\tau(a_\tau^d,h_\tau^d)
\prod_{k=1}^{\tau}
dP\!\left(
s_k,l_k
\mid
A_{k-1}=a_{k-1}^d,
H_{k-1}=h_{k-1}^d
\right).
\]
The right-hand side is exactly the identifying functional in Theorem 1
and therefore equals \(\theta\). Hence, $
\theta=\E\{q_1(A_1,H_1)\}$.
\end{proof}
}

\section{Efficient influence functions (Corollary \ref{coro:eif})}
\begin{proof}
  In this proof we will use $\Theta(\P)$ to denote a parameter
  as a functional that maps the distribution $\P$ in the model to a
  real number.  A function $\D_1(Z;\eta)$ is the EIF of a
  parameter functional $\Theta(P)$ if it satisfies
  \begin{equation}
    \frac{\dd}{\dd  \epsilon}\Theta(P_\epsilon)\bigg|_{\epsilon=0} =
    \E[\D_1(Z;\eta) s(Z)],\label{eq:defeif}
  \end{equation}
  where $\P_\epsilon$ is a smooth parametric submodel that locally
  covers the non-parametric model, with score
  \[s(Z)=\left(\frac{\dd \log \P_\epsilon}{\dd
        \epsilon}\right)\bigg|_{\epsilon = 0}\] such that
  $\P_{\epsilon=0}=\P$. We start by conjecturing an EIF, and then
  prove it is in fact the EIF using the above definition.
  
  To conjecture an EIF we find the EIF in a model where we assume that
  the measure $\nu$ is discrete so that integrals can be written as
  sums.  Recall that for fixed natural treatment values $\bar s_t$ and
  covariate values $\bar l_t$, we
  denote $a_t^\d=\d_t(\bar s_t, h^\d_t)$, where
  $h^\d_t=(\bar a_{t-1}^\d, \bar l_t)$. We can express the
  non-parametric MLE of $\theta$ as follows, writing $L_{\tau+1}=Y$.
  All sums below range over natural treatment inputs and covariate
  values; assigned treatments and histories are determined recursively:
  \begin{equation}
    \Theta(\Pn)=\sum_{\bar s_\tau, \bar l_{\tau+1}}l_{\tau+1}\Pn
    f_{l_{\tau+1},a_\tau^d,h_\tau^d}\prod_{k=1}^\tau\frac{\Pn f_{s_k,
        l_k, a_{k-1}^d, h_{k-1}^d}}{\Pn f_{a_k^d, h_k^d}}\label{nonpest},
  \end{equation}
  where we remind the reader of the notation $\P f =\int f
  \dd\P$. Here $f_{x}(X)=\one(X=x)$ for any variable $X$ and value
  $x$.  In the above display we used the assumption that $\d$ does not
  depend on $\P$, and therefore does not need to be estimated.

  We will use the fact that the efficient influence function in a
  non-parametric model corresponds with the influence curve of the
  NPMLE. This is true because the influence curve of any regular
  estimator is also a gradient, and a non-parametric model has only
  one gradient. The Delta method shows that if $\hat \Theta(\Pn)$ is a
  substitution estimator such that $\theta=\hat \Theta(\P)$, and
  $\hat \Theta(\Pn)$ can be written as
  $\hat \Theta^*(\Pn f:f\in\mathcal{F})$ for some class of functions
  $\mathcal{F}$ and some mapping $\Theta^*$, then the influence
  function of $\hat \Theta(\Pn)$ is equal to
  \[\dr_\P(Z)=\sum_{f\in\mathcal{F}}\frac{\dd\hat \Theta^*}{\dd\P f}\{f(Z)-\P f\}.\]

  Applying this result to (\ref{nonpest}) with
  $\mathcal{F}=\{f_{l_{\tau+1},a_\tau^d,h_\tau^d}, f_{s_\ell, l_\ell, a_{\ell-1}^d, h_{\ell-1}^d},
  f_{a_\ell^d, h_\ell^d}: \bar s_\tau, \bar l_{\tau+1}; \ell=1,\ldots,\tau\}$ and
  rearranging terms gives $\dr_P(Z) = \D_1(Z;\eta)$. In particular we have
  \begin{align*}
    \frac{\dd \Theta(\P)}{\dd \P f_{l_{\tau+1},a_\tau^d,h_\tau^d}} &=
                                                              l_{\tau+1}\prod_{k=1}^\tau\frac{\P f_{s_k,
                                                              l_k, a_{k-1}^d, h_{k-1}^d}}{\P f_{a_k^d, h_k^d}},
  \end{align*}
  and 
{\small  \begin{align*}
    &\frac{\dd \Theta(\P)}{\dd \P f_{s_\ell,
    l_\ell, a_{\ell-1}^d, h_{\ell-1}^d}}
    = \sum_{l_{\tau+1},
      \underline s_{\ell+1},
      \underline l_{\ell+1}}l_{\tau+1}\frac{\P
      f_{l_{\tau+1},a_\tau^d,h_\tau^d}}{\P f_{s_\ell,
      l_\ell, a_{\ell-1}^d, h_{\ell-1}^d}}\prod_{k=1}^\tau\frac{\P f_{s_k,
      l_k, a_{k-1}^d, h_{k-1}^d}}{\P f_{a_k^d, h_k^d}}\\
    &= \left(\sum_{\underline s_{\ell+1},
      \underline l_{\ell+1}}m_\tau(a_\tau^d, h_\tau^d)\frac{\P
      f_{a_\tau^d,h_\tau^d}}{\P
      f_{a_\ell^d,h_\ell^d}}\prod_{k=\ell+1}^\tau\frac{\P f_{s_k,
      l_k, a_{k-1}^d, h_{k-1}^d}}{\P f_{a_k^d, h_k^d}}\right)\\ &\quad{}\times\prod_{k=1}^{\ell-1}\frac{\P f_{s_k,
      l_k, a_{k-1}^d, h_{k-1}^d}}{\P f_{a_k^d, h_k^d}}\\
    &= \left(\sum_{\substack{s_{\ell+1},\ldots,s_{\tau-1}\\
      l_{\ell+1}, \ldots, l_{\tau-1}}}\left[\sum_{s_\tau, l_\tau}m_\tau(a_\tau^d, h_\tau^d)\frac{\P f_{s_\tau,
      l_\tau, a_{\tau-1}^d, h_{\tau-1}^d}}{\P f_{a_{\tau-1}^d, h_{\tau-1}^d}}\right]\frac{\P
      f_{a_{\tau-1}^d,h_{\tau-1}^d}}{\P
      f_{a_\ell^d,h_\ell^d}}\prod_{k=\ell+1}^{\tau-1}\frac{\P f_{s_k,
      l_k, a_{k-1}^d, h_{k-1}^d}}{\P f_{a_k^d, h_k^d}}\right)\\ &\quad{}\times\prod_{k=1}^{\ell-1}\frac{\P f_{s_k,
    l_k, a_{k-1}^d, h_{k-1}^d}}{\P f_{a_k^d, h_k^d}}\\
    &= \left(\sum_{\substack{s_{\ell+1},\ldots,s_{\tau-1}\\
    l_{\ell+1}, \ldots, l_{\tau-1}}}\left[m_{\tau-1}(\bar s_{\tau-1}, a_{\tau-1}^d, h_{\tau-1}^d)\right]\frac{\P
    f_{a_{\tau-1}^d,h_{\tau-1}^d}}{\P
    f_{a_\ell^d,h_\ell^d}}\prod_{k=\ell+1}^{\tau-1}\frac{\P f_{s_k,
    l_k, a_{k-1}^d, h_{k-1}^d}}{\P f_{a_k^d, h_k^d}}\right)\\ &\quad{}\times\prod_{k=1}^{\ell-1}\frac{\P f_{s_k,
    l_k, a_{k-1}^d, h_{k-1}^d}}{\P f_{a_k^d, h_k^d}}\\
    &= m_\ell(\bar s_\ell, a_\ell^d, h_\ell^d)\prod_{k=1}^{\ell-1}\frac{\P f_{s_k,
      l_k, a_{k-1}^d, h_{k-1}^d}}{\P f_{a_k^d, h_k^d}}, 
  \end{align*}}
  as well as
  \begin{align*}
    \frac{\dd \Theta(\P)}{\dd \P f_{a_\ell^d,h_\ell^d}} &= -\sum_{l_{\tau+1},
                                                \underline s_{\ell+1},
                                                \underline l_{\ell+1}}l_{\tau+1}\frac{\P
                                                    f_{l_{\tau+1},
                                                    a_\tau^d, h_\tau^d}}{\P f_{a_\ell^d, h_\ell^d}}\prod_{k=1}^\tau\frac{\P f_{s_k,
                                                    l_k, a_{k-1}^d, h_{k-1}^d}}{\P f_{a_k^d, h_k^d}}\\
    &=-\sum_{l_{\tau+1},
    \underline s_{\ell+1},
    \underline l_{\ell+1}}l_{\tau+1}\frac{\P
    f_{l_{\tau+1},
    a_\tau^d, h_\tau^d}}{\P f_{a_\ell^d, h_\ell^d}}\prod_{k=\ell+1}^\tau\frac{\P f_{s_k,
    l_k, a_{k-1}^d, h_{k-1}^d}}{\P f_{a_k^d, h_k^d}}\prod_{k=1}^\ell\frac{\P f_{s_k,
      l_k, a_{k-1}^d, h_{k-1}^d}}{\P f_{a_k^d, h_k^d}}\\
                                                  &= -m_\ell(\bar s_\ell, a_\ell^d, h_\ell^d)\prod_{k=1}^{\ell}\frac{\P f_{s_k,
                                                    l_k, a_{k-1}^d, h_{k-1}^d}}{\P f_{a_k^d, h_k^d}}.
  \end{align*}
  Using the policy indicators $D_{t,k}$ defined in the main paper,
  notice now that
  {\small \begin{align*}
            \sum_{\bar s_\tau,\bar l_{\tau+1}}&\frac{\dd \Theta(\P)}{\dd
                                             f_{l_{\tau+1},a_\tau^d,h_\tau^d}}
                                             f_{l_{\tau+1},a_\tau^d,h_\tau^d}(Z) \\
                                           &=\sum_{\bar s_\tau,\bar l_\tau}\one\{A_\tau
                                             = a_\tau^d, H_\tau =
                                       h_\tau^d\}L_{\tau+1}\prod_{k=1}^\tau\frac{\P(A_k=s_k,L_k=
                                       l_k, A_{k-1}=a_{k-1}^d,
                                       H_{k-1}=h_{k-1}^d)}{\P(A_k=a_k^d,
                                       H_k=h_k^d)}\\
                                     &=L_{\tau+1}\sum_{\bar
                                       s_\tau}\one\{A_\tau = d_\tau(\bar s_\tau,H_\tau),\ldots,
                                       A_1=d_1(s_1,L_1)\}\prod_{k=1}^\tau\frac{g_k(s_k
                                       \mid H_k)}{g_k(A_k\mid
                                       H_k)}\\
                                     &=L_{\tau+1}\sum_{\bar s_\tau}
                                       D_{0,\tau}(\bar s_\tau)
                                       \prod_{k=1}^\tau
                                       \frac{g_k(s_k\mid H_k)}{g_k(A_k\mid H_k)},
         \end{align*}}
  as well as
  {\small  
\begin{align*}
&\sum_{\bar s_\ell,\bar l_\ell}\frac{\dd \Theta(\P)}{\dd \P f_{s_\ell,l_\ell, a_{\ell-1}^d,
            h_{\ell-1}^d}} f_{s_\ell, l_\ell, a_{\ell-1}^d, h_{\ell-1}^d}(Z) \\
        &=\sum_{\bar s_{\ell-1}}\one\{A_{\ell-1} = d_{\ell-1}(\bar s_{\ell-1},H_{\ell-1}),\ldots, A_1=d_1(s_1,L_1)\}\\ &\quad{}\times m_\ell((\bar s_{\ell-1}, A_\ell),
                             d_\ell((\bar s_{\ell-1},A_\ell),H_\ell), H_\ell)\prod_{k=1}^{\ell-1}\frac{g_k(s_k
                             \mid
                             H_k)}{g_k(A_k\mid 
                             H_k)},
           \end{align*}}%
and
{\small
  \begin{align*}
    \sum_{\bar s_\ell,\bar l_\ell}&\frac{\dd \Theta(\P)}{\dd \P f_{a_\ell^d,
                    h_\ell^d}} f_{a_\ell^d, h_\ell^d}(Z) \\
                  &=-\sum_{\bar s_\ell}\one\{A_\ell = d_\ell(\bar s_\ell,H_\ell),\ldots, A_1=d_1(s_1,L_1)\}m_\ell(\bar s_\ell,
                    A_\ell, H_\ell)\prod_{k=1}^\ell\frac{g_k(s_k
                    \mid
                    H_k)}{g_k(A_k\mid 
                    H_k)},
  \end{align*}}%
Putting these results together and summing over $\ell$ yields the
conjectured $\dr_P(Z)=\D_1(Z;\eta)-\Theta(P)$. Then, we can confirm that
$\D_1(Z;\eta)$ satisfies (\ref{eq:defeif}). To see this, note that
\edit{Theorem~\ref{theo:sdr}} implies that the functional satisfies the
following expansion:
\[ \Theta(\P_\epsilon) = \Theta(\P) + \int \{\D_1(Z;\eta) - \Theta(\P)\} \dd\P_\epsilon - \rem_0(\eta_\epsilon,\eta),\]
  Differentiating with respect to $\epsilon$ and evaluating at
  $\epsilon=0$ yields
  \begin{align*}
    \frac{\dd }{\dd \epsilon}\Theta(\P_\epsilon)\bigg|_{\epsilon = 0}
    &= \int\{\D_1(Z;\eta) - \Theta(\P)\}
      \left(\frac{\dd  \P_\epsilon}{\dd \epsilon}\right)\bigg|_{\epsilon
      = 0} - \frac{\dd }{\dd
      \epsilon}\rem_0(\eta_\epsilon,\eta)\bigg|_{\epsilon = 0} \\
    &= \int\{\D_1(Z;\eta) - \Theta(\P)\}
      \left(\frac{\dd \log \P_\epsilon}{\dd
      \epsilon}\right)\bigg|_{\epsilon = 0}\dd P - \frac{\dd }{\dd \epsilon}\rem_0(\eta_\epsilon,\eta)\bigg|_{\epsilon = 0},
  \end{align*}
  and the result follows after noticing that
  \[\frac{\dd }{\dd \epsilon}\rem_0(\eta_\epsilon,\eta)\bigg|_{\epsilon = 0}=0.\]
\end{proof}

\section{Sequential double robustness (\edit{Theorem \ref{theo:sdr}})}
We first prove an auxiliary result.
\begin{proposition} \edit{For any function $m_{k,1}(\bar s_k, a_k, h_k)$ define $D_k(\bar s_k)= \one\{A_k=d_k(\bar s_k,
           H_k)\}$, and 
  \[q_{k,1}(\bar s_{k-1}, A_k, H_k)=m_{k,1}((\bar s_{k-1}, A_k), d_k((\bar s_{k-1}, A_k), H_k),
  H_k).\] We have {\small \begin{multline} \E[q_{k,1}(\bar s_{k-1}, A_k,
      H_k)\mid A_{k-1}=a_{k-1}, H_{k-1}=h_{k-1}]=\\
      \E\left[\sum_{s_k}D_k(\bar s_k)\frac{g_k(s_k\mid
          H_k)}{g_k(A_k\mid H_k)} m_{k, 1}(\bar s_k, A_k, H_k)\mid A_{k-1}=a_{k-1},
        H_{k-1}=h_{k-1}\right].\end{multline}}%
}
\end{proposition}
\edit{In the calculation below, $a_k$ indexes the observed treatment
coordinate separately from the natural input $s_k$.}
\begin{proof}
\edit{
  The right hand side can be written as
  {\small\begin{align*}
           \sum_{a_k, l_k}&\sum_{s_k}\one\{a_k=d_k(\bar s_k,
           h_k)\}\frac{g_k(s_k\mid
           h_k)}{g_k(a_k\mid h_k)}
           m_{k,1}(\bar s_{k}, a_{k},
           h_{k})g_k(a_k\mid h_k)p(l_k\mid
           a_{k-1}, h_{k-1})\\
           &=\sum_{a_{k}, l_{k}}\sum_{s_{k}}\one\{a_{k}=d_k(\bar  s_{k},
             h_{k})\}g_{k}(s_{k}\mid
             h_{k}) m_{k,1}(\bar s_{k}, a_{k},
             h_{k})p(l_{k}\mid
             a_{k-1}, h_{k-1})\\
           &=\sum_{l_{k}, s_{k}}g_{k}(s_{k}\mid
             h_{k}) m_{k,1}(\bar s_k, d_k(\bar s_{k},
             h_{k}),
             h_{k})p(l_{k}\mid
             a_{k-1}, h_{k-1})\\
           &=\sum_{l_{k}, s_{k}}g_{k}(s_{k}\mid
             h_{k}) q_{k,1}(\bar s_{k-1}, s_{k},
             h_{k})p(l_{k}\mid
             a_{k-1}, h_{k-1})\\
           &=     \E[q_{k,1}(\bar s_{k-1}, A_k,
             H_k)\mid A_{k-1}=a_{k-1}, H_{k-1}=h_{k-1}].    
         \end{align*}}
}
     \end{proof}
     We now proceed with the proof of the main result (\edit{Theorem~\ref{theo:sdr}}).
     
\begin{proof}
  This \edit{theorem} follows from recursive application for
  $k=t+1,\ldots,\tau$ of the following
  relation:
{\footnotesize \begin{align*}
    &\textcolor{black}{\E[q_{k,1}}\textcolor{black}{(\bar s_{k-1}, A_k, H_k) -
                                  q_k(\bar s_{k-1}, A_k, H_k)\mid A_{k-1},
    H_{k-1}]} =\\
    &\E\left\{\sum_{s_k}D_k(\bar s_k)\left[\frac{\g_k(s_k\mid
       H_k)}{\g_k(A_k\mid H_k)}-\frac{\g_{k,1}(s_k\mid
       H_k)}{\g_{k,1}(A_k\mid H_k)}\right][m_{k,1}(\bar s_{k}, A_k,
       H_k) - m_k(\bar s_{k}, A_k,
       H_k)]\,\,\bigg|\,\,A_{k-1},H_{k-1}\right\}\\
    -&\E\left\{\sum_{s_k}D_k(\bar s_k)\frac{\g_{k,1}(s_k\mid
               H_k)}{\g_{k,1}(A_k\mid H_k)}[q_{k+1, 1}(\bar s_k, A_{k+1},
         H_{k+1}) - m_{k, 1}(\bar s_k, A_k,
         H_k)]\,\,\bigg|\,\,A_{k-1},H_{k-1}\right\}\\
    +&\E\left\{\sum_{s_k}D_k(\bar s_k)\frac{\g_{k,1}(s_k\mid
       H_k)}{\g_{k,1}(A_k\mid H_k)}\textcolor{black}{\E[q_{k+1, 1}(\bar s_k,
       A_{k+1},
       H_{k+1}) - q_{k+1}(\bar s_k, A_{k+1},
       H_{k+1})\mid A_k, H_k]}\,\,\bigg|\,\,A_{k-1},H_{k-1}\right\},
  \end{align*}}
which follows from rewriting it as
\begin{align*}
  a-b &= \sum (c-d)(a_1 - b_1)\\
      &- \sum d(a_2-b_2)\\
      &+ \sum d(a_3-b_3),
\end{align*}
after noticing that $\sum d a_2=\sum da_3$, $\sum d b_2=\sum da_1$,
$\sum d b_1=\sum db_3$, $a=\sum c a_1$, and $b=\sum c b_1$, where the
latter two follow from the above proposition.
\end{proof}

\section{Augmented-data implementation}
\label{app:pooled-implementation}
A na\"ive plug-in estimator for $\Q_t:t=1,\ldots,\tau$ may be obtained using
%Proposition~\ref{propo:sr} as follows. Start by running a regression
%of $q_{\tau+1} = \Q_{\tau+1} = Y_i$ on $(A_{\tau,i},H_{\tau,i})$ to
%obtain an estimator of $m_\tau(A_\tau, H_\tau)$. Then, for each
%$\bar s_{\tau-1}$, compute the predictions
%$\hat q_{\tau}(\bar s_{\tau-1}, A_\tau, H_\tau)$ according to the definitions in the proposition. For each
%$\bar s_{\tau-1}$, use these predictions as a pseudo-outcome in a
%regression on $(A_{\tau-1}, H_{\tau-1})$, to obtain an estimate
%$\hat\Q_{\tau-1}(\bar s_{\tau-1}, a_{\tau-1}, h_{\tau-1})$ and iterate the
%process until obtaining an estimate $\hat\Q_1(s_1, a_1, h_1)$. 
%Alternatively, $\Q_t:t=1,\ldots,\tau$ may be estimated using
pooled regressions in an augmented dataset. \edit{For the
  augmented-data implementation described below, we additionally
  assume that $A_t$ has finite support at each time point.} Let
$K_t = |\bar{\mathcal{A}}_t|$. Define the augmented dataset
$\mathcal{D}_\tau = \left\{(i, Z_i, \edit{\bar s_{\tau-1}}) : i \in \{1, \dots,
  n\},\; \bar s_{\tau-1} \in \bar{\mathcal{A}}_{\tau-1} \right\}$. That is,
$\mathcal{D}_\tau$ is a dataset with $n K_{\tau-1}$ rows and $\tau-1$
additional columns corresponding to \edit{possible sequences of natural
treatment values} $\bar s_{\tau-1}$. Start by running a regression of $Y_i$ on
$(A_{\tau,i},H_{\tau,i})$ to obtain an estimator of
$m_\tau(a_\tau, h_\tau) = E[Y \mid A_\tau = a_\tau, H_\tau = h_\tau]$
among the original data. Generate predictions
$q_\tau(\bar s_{\tau-1}, A_{\tau}, H_{\tau}) = m_\tau(d_\tau((\bar
s_{\tau-1}, A_\tau), H_\tau), H_\tau) = E[Y \mid A_\tau = d_\tau((\bar
s_{\tau-1}, A_\tau), H_\tau), H_\tau]$ using the augmented data
$\mathcal{D}_\tau$. Because predictions are generated using the
augmented data, $q_\tau(\bar s_{\tau-1}, A_{\tau}, H_{\tau})$ is a
$n K_{\tau-1}$ length vector. Recursively for
$t=\tau-1,\ldots, 1$: \begin{enumerate}[label=(\arabic*),leftmargin=*]
    \item Add $q_{t+1}(\bar s_t,A_{t+1},H_{t+1})$ as a new
    column to $\mathcal D_{t+1}$.
    \item \edit{Using $\mathcal D_{t+1}$, retain all columns of
    $\bar s_t$ and estimate
    \[
    m_t(\bar s_t,a_t,h_t)
    =\E\{q_{t+1}(\bar s_t,A_{t+1},H_{t+1})\mid A_t=a_t,H_t=h_t\}
    \]
    by regressing the appended pseudo-outcome on $(\bar s_t,A_t,H_t)$.}
    \item \edit{Generate predictions by setting the synthetic argument
    $s_t=A_t$ and evaluating the treatment argument at
    $d_t((\bar s_{t-1},A_t),H_t)$:
    \[
    q_t(\bar s_{t-1},A_t,H_t)
    =m_t\{(\bar s_{t-1},A_t),d_t((\bar s_{t-1},A_t),H_t),H_t\}.
    \]}
    \item \edit{After generating these predictions, construct
    $\mathcal D_t$ by retaining unique combinations of
    $(i,\bar s_{t-1})$, together with $q_t$. The resulting dataset
    has $nK_{t-1}$ rows.}
\end{enumerate}
An example of this approach is shown in Figure~\ref{fig:pooled}. 

Note that this approach replaces $1+\sum_{t=1}^{\tau-1}K_t$ separate
regressions in size-$n$ datasets with $\tau$ pooled regressions in
datasets of sizes $n,nK_{\tau-1},\ldots,nK_1$, respectively.  This
pooling strategy may improve computational efficiency. A similar
pooling strategy has also previously been used for estimating
so-called time-smooth counterfactual curves
\citep{susmann2025computationally}.  Here, we give the algorithm in
full generality for interventions that may depend on the entire
history $\bar s_t$. When the intervention depends only on selected
components of $\bar s_t$, the augmented data need only carry those
components. For example, delay policies require only a rolling lagged
window of natural treatment values \edit{whose length is determined by the}
magnitude of the delay. We give concrete examples in
\S\ref{app:example}.
\begin{figure}[H]
\footnotesize
\textit{Step 1.} Estimate $m_3(a_3,a_2,a_1)$ by regressing $Y$ on $(A_3,A_2,A_1)$ and create the augmented dataset.
\[
\begin{array}{cccccc}
i & A_1 & A_2 & A_3 & Y \\
\hline
1 & a_{1,1} & a_{1,2} & a_{1,3} & y_1\\
\vdots & \vdots & \vdots & \vdots & \vdots \\
n & a_{n,1} & a_{n,2} & a_{n,3} & y_n\\
\end{array} \mapsto
\begin{array}{ccccccccc}
i & A_1 & A_2 & A_3 & s_1 & s_2 & Y \\
\hline
1 & a_{1,1} & a_{1,2} & a_{1,3} & 0 & 1 & y_1\\
1 & a_{1,1} & a_{1,2} & a_{1,3} & 0 & 0 & y_1\\
1 & a_{1,1} & a_{1,2} & a_{1,3} & 1 & 1 & y_1\\
1 & a_{1,1} & a_{1,2} & a_{1,3} & 1 & 0 & y_1\\
\vdots & \vdots & \vdots & \vdots & \vdots & \vdots & \vdots & \vdots \\
n & a_{n,1} & a_{n,2} & a_{n,3} & 1 & 0 & y_n\\
\end{array}
\]
\textit{Step 2.} Generate predictions \edit{$q_3(\bar s_2,A_3)=m_3(d_3((\bar s_2,A_3),H_3),A_2,A_1)$} in the augmented dataset and subset to the unique rows of $(i, s_1, s_2)$. \edit{Estimate $m_2(\bar s_2,A_2,A_1)$ by regressing $q_3(\bar s_2,A_3)$ on $(s_1,s_2,A_2,A_1)$.}
\[
\begin{array}{ccccccc}
i & A_1 & A_2 & s_1 & s_2 & q_3(\bar{s}_2, A_3) \\
\hline
1 & a_{1,1} & a_{1,2} & 0 & 1 & \edit{q_3((0,1), a_{1,3})}\\
1 & a_{1,1} & a_{1,2} & 0 & 0 & \edit{q_3((0,0), a_{1,3})}\\
1 & a_{1,1} & a_{1,2} & 1 & 1 & q_3((1,1), a_{1,3})\\
1 & a_{1,1} & a_{1,2} & 1 & 0 & q_3((1,0), a_{1,3})\\
\vdots & \vdots & \vdots & \vdots & \vdots & \vdots \\
n & a_{n,1} & a_{n,2} & 1 & 0 & q_3((1,0), a_{n,3})
\end{array}
\]
\textit{Step 3.} Generate predictions \edit{$q_2(s_1,A_2)=m_2((s_1,A_2),d_2((s_1,A_2),H_2),A_1)$} in the augmented dataset and subset to the unique rows of $(i, s_1)$. Estimate $m_1(s_1,A_1)$ by regressing $q_2(s_1, A_2)$ on $(s_1,A_1)$.
\[
\begin{array}{cccccccc}
i & A_1 & A_2 & s_1 & s_2 & q_3(\bar{s}_2, A_3) & q_2(s_1, A_2) \\
\hline
1 & a_{1,1} & a_{1,2} & 0 & 1 & \edit{q_3((0,1), a_{1,3})} & q_2(0, a_{1,2})\\
1 & a_{1,1} & a_{1,2} & 0 & 0 & \edit{q_3((0,0), a_{1,3})} & q_2(0, a_{1,2})\\
1 & a_{1,1} & a_{1,2} & 1 & 1 & q_3((1,1), a_{1,3}) & q_2(1, a_{1,2})\\
1 & a_{1,1} & a_{1,2} & 1 & 0 & q_3((1,0), a_{1,3}) & q_2(1, a_{1,2})\\
\vdots & \vdots & \vdots & \vdots & \vdots & \vdots & \vdots \\
n & a_{n,1} & a_{n,2} & 1 & 0 & q_3((1,0), a_{n,3}) & q_2(1, a_{n,2})\\
\end{array} \mapsto 
\begin{array}{cccccccc}
i & A_1 & s_1 & q_2(s_1, A_2) \\
\hline
1 & a_{1,1} & 0 & q_2(0, a_{1,2})\\
1 & a_{1,1} & 1 & q_2(1, a_{1,2})\\
\vdots & \vdots & \vdots & \vdots \\
n & a_{n,1} & 1 & q_2(1, a_{n,2})\\
\end{array}
\]
\textit{Step 4.} Generate predictions \edit{$q_1(A_1)=m_1(A_1,d_1(A_1,H_1))$} in the augmented dataset and subset to the unique rows of $i$.
\[
\begin{array}{cccccccc}
i & A_1 & s_1 & q_2(s_1, A_2) & q_1(A_1) \\
\hline
1 & a_{1,1} & 0 & q_2(0, a_{1,2}) & q_1(a_{1,1})\\
1 & a_{1,1} & 1 & q_2(1, a_{1,2}) & q_1(a_{1,1})\\
\vdots & \vdots & \vdots & \vdots & \vdots \\
n & a_{n,1} & 1 & q_2(1, a_{n,2}) & q_1(a_{n,1})\\
\end{array} 
\mapsto
\begin{array}{cccccccc}
i & q_1(A_1) \\
\hline
1 &  q_1(a_{1,1})\\
\vdots & \vdots  \\
n & q_1(a_{n,1})\\
\end{array} 
\]
\caption{Schematic of the pooled-regression approach to Proposition~\ref{propo:sr} with $\tau = 3$, $A_t \in \{0,1\}$, and no covariates. \edit{Here $H_t=\bar A_{t-1}$, with $H_1$ empty. The entries $a_{i,t}$ are observed treatment values; the separate $s_j$ columns index retained natural inputs. Past treatment coordinates are written out in the regression arguments, and the corresponding history arguments of $q_t$ are suppressed.}}
\label{fig:pooled}
\end{figure}
\section{Delay intervention with absorbing states}\label{app:example}

Proposition~\ref{propo:sr} shows that the augmented
history $\bar s_t$ need not be carried in full when the intervention
depends only on selected components of the natural treatment history.
\edit{The vaccine example in Section~\ref{sec:delay} of the main
paper concerns receipt of treatment at a particular time. The
formulation may differ for an absorbing exposure, such as the status
of having received a vaccine booster by time $t$. The sequence
$\{A_1,A_2(A_1^d),A_3(\bar A_2^d),\ldots\}$ need not be absorbing
even when the observed treatment process is absorbing. For example,
consider $\tau=3$ with no covariates, and suppose both initiation and
remaining untreated have positive probability at each time before
treatment begins. Under a one-timepoint lag, the natural inputs
$(s_1,s_2,s_3)=(1,0,1)$ are reachable: after $a_1^d=0$, the next
natural value $s_2=A_2(0)$ can equal zero; after $a_2^d=s_1=1$,
absorption implies $s_3=A_3(0,1)=1$. The plain lag rule
$a_t^d=s_{t-1}$ for $t>1$ then assigns $(0,1,0)$. However,
$P(A_3=0\mid A_1=0,A_2=1)=0$, so the assigned treatment--history
pair at time $3$ is outside the observed-data support, violating
Assumption~\ref{ass:support}.

To delay initiation by $\kappa$ timepoints while preserving an
absorbing treatment status, define
\[
d_t^{(\kappa)}(\bar s_t,h_t)=
\begin{cases}
0, & t\leq\kappa,\\
\max\{a_{t-1},s_{t-\kappa}\}, & t>\kappa,
\end{cases}
\qquad h_t=(\bar a_{t-1},\bar l_t).
\]
The maximum preserves treatment once it has been initiated. The
rule initiates treatment $\kappa$ timepoints after the first natural
indication generated under the policy and maintains treatment
thereafter, subject to the positivity assumption. It needs only the
lagged natural value $s_{t-\kappa}$ and the previous modified
treatment already included in $h_t$.}
We discuss two examples.

\begin{example}[Delay by one time point]\label{ex:delay-one}
\edit{For $\tau=3$, the rule is
$d_1(s_1,h_1)=0$,
$d_2(\bar s_2,h_2)=\max\{a_1,s_1\}$, and
$d_3(\bar s_3,h_3)=\max\{a_2,s_2\}$.
The sequential regression formula becomes
\begin{align*}
  m_3(A_3,H_3)&=\E[Y\mid A_3,H_3],\\
  q_3(s_2,A_3,H_3)&=m_3(\max\{A_2,s_2\},H_3),\\
  m_2(s_2,A_2,H_2)&=\E\{q_3(s_2,A_3,H_3)\mid A_2,H_2\},\\
  q_2(s_1,A_2,H_2)&=m_2(A_2,\max\{A_1,s_1\},H_2),\\
  m_1(s_1,A_1,H_1)&=\E\{q_2(s_1,A_2,H_2)\mid A_1,H_1\},\\
  q_1(A_1,H_1)&=m_1(A_1,0,H_1),\\
  \theta&=\E\{q_1(A_1,H_1)\}.
\end{align*}
In $q_3$ and $m_2$, the unused argument $s_1$ is suppressed.
After substitution at time $2$, $q_2$ depends on $s_1$ and $A_2$;
joint augmentation by $(s_1,s_2)$ is unnecessary. Only one natural
treatment value is carried at each recursion step, as illustrated
in Figure~\ref{fig:pooled-delay-one}.}

This example allows us to also illustrate the canonical gradients, which are equal to
  {\footnotesize  \begin{align*}
                    \psi_3(s_2, Z)&=\sum_{s_3}D_{2,3}(s_2, s_3)\frac{g_3(s_3\mid
                                    H_3)}{g_3(A_3\mid H_3)}(Y - m_3(A_3, H_3))\\
                                  &+
                                    q_3(s_2, A_3, H_3) - m_2( s_2, A_2, H_2)\\
                    \psi_2(s_1, Z)&=\sum_{s_2,s_3}D_{1,3}(\bar s_3)\frac{g_2(s_2\mid
                                    H_2)}{g_2(A_2\mid H_2)}\frac{g_3(s_3\mid
                                    H_3)}{g_3(A_3\mid H_3)}(Y - m_3(A_3, H_3))\\
                                  &+\sum_{s_2}D_{1,2}(\bar s_2)\frac{g_2(s_2\mid
                                    H_2)}{g_2(A_2\mid H_2)}(q_3(s_2, A_3, H_3) -
                                    m_2(s_2, A_2, H_2))\\
                                  &+q_2(s_1, A_2, H_2) - m_1(s_1, A_1, H_1)\\
                    \psi_1(Z)&= \sum_{s_1, s_2, s_3}D_{0,3}(\bar s_3)\frac{g_1(s_1\mid
                               H_1)}{g_1(A_1\mid H_1)}\frac{g_2(s_2\mid
                               H_2)}{g_2(A_2\mid H_2)}\frac{g_3(s_3\mid
                               H_3)}{g_3(A_3\mid H_3)}(Y - m_3(A_3, H_3))\\
                                  &+\sum_{s_1,s_2}D_{0,2}(\bar s_2)\frac{g_1(s_1\mid
                                    H_1)}{g_1(A_1\mid H_1)}\frac{g_2(s_2\mid
                                    H_2)}{g_2(A_2\mid H_2)}(\edit{q_3(s_2, A_3, H_3)} -
                                    \edit{m_2(s_2, A_2, H_2)})\\
                                  &+\sum_{s_1}D_{0,1}(s_1)\frac{g_1(s_1\mid
                                    H_1)}{g_1(A_1\mid H_1)}(\edit{q_2(s_1, A_2, H_2)} -
                                    m_1(s_1, A_1, H_1))\\
                                  &+ q_1(A_1, H_1) - \theta.
                  \end{align*}   }   
              \end{example}    
\begin{figure}[H]
\scriptsize
\textit{Step 1.} Estimate $m_3(a_3,h_3)$ by regressing $Y$ on
$(A_3,H_3)$. Since $q_3$ depends only on $s_2$, create an augmented
dataset with one additional column $s_2$.
\[
\begin{array}{ccccc}
i & A_1 & A_2 & A_3 & Y \\
\hline
1 & a_{1,1} & a_{1,2} & a_{1,3} & y_1\\
\vdots & \vdots & \vdots & \vdots & \vdots \\
n & a_{n,1} & a_{n,2} & a_{n,3} & y_n\\
\end{array}
\mapsto
\begin{array}{cccccc}
i & A_1 & A_2 & A_3 & s_2 & Y \\
\hline
1 & a_{1,1} & a_{1,2} & a_{1,3} & 0 & y_1\\
1 & a_{1,1} & a_{1,2} & a_{1,3} & 1 & y_1\\
\vdots & \vdots & \vdots & \vdots & \vdots & \vdots \\
n & a_{n,1} & a_{n,2} & a_{n,3} & 1 & y_n\\
\end{array}
\]
\textit{Step 2.} Generate \edit{$q_3(s_2,A_3)=m_3(\max\{A_2,s_2\},H_3)$} and estimate
$m_2(s_2,A_2,H_2)$ by regressing $q_3(s_2,A_3)$ on
$(s_2,A_2,H_2)$.
\[
\begin{array}{ccccc}
i & A_1 & A_2 & s_2 & q_3(s_2,A_3) \\
\hline
1 & a_{1,1} & a_{1,2} & 0 & q_3(0,a_{1,3})\\
1 & a_{1,1} & a_{1,2} & 1 & q_3(1,a_{1,3})\\
\vdots & \vdots & \vdots & \vdots & \vdots \\
n & a_{n,1} & a_{n,2} & 1 & q_3(1,a_{n,3})
\end{array}
\]
\textit{Step 3.} To estimate $m_1$, augment the data with an additional column $s_1$. Generate
\edit{$q_2(s_1,A_2)=m_2(A_2,\max\{A_1,s_1\},H_2)$} and regress it on
$(s_1,A_1,H_1)$.
\[
\begin{array}{ccccc}
i & A_1 & A_2 & s_1 & q_2(s_1,A_2) \\
\hline
1 & a_{1,1} & a_{1,2} & 0 & q_2(0,a_{1,2})\\
1 & a_{1,1} & a_{1,2} & 1 & q_2(1,a_{1,2})\\
\vdots & \vdots & \vdots & \vdots & \vdots \\
n & a_{n,1} & a_{n,2} & 1 & q_2(1,a_{n,2})
\end{array}
\mapsto
\begin{array}{cc}
i & q_1(A_1) \\
\hline
1 & q_1(a_{1,1})\\
\vdots & \vdots \\
n & q_1(a_{n,1})
\end{array}
\]
\caption{Reduced pooled-regression implementation of the one-timepoint
delay in Example~\ref{ex:delay-one}, with $\tau=3$, $A_t\in\{0,1\}$,
and no covariates. Unlike the fully general implementation in
Figure~\ref{fig:pooled}, this intervention never requires simultaneous
augmentation by $(s_1,s_2)$: only the single natural value that may be
used in a future delayed treatment decision is carried at each step.}
\label{fig:pooled-delay-one}
\end{figure}
              \begin{example}[Delay by two time points]\label{ex:delay-two}
\edit{For a two-timepoint delay with $\tau=3$, the rule is
$d_1(s_1,h_1)=0$, $d_2(\bar s_2,h_2)=0$, and
$d_3(\bar s_3,h_3)=\max\{a_2,s_1\}$.
The sequential regression formula becomes
\begin{align*}
  m_3(A_3,H_3)&=\E[Y\mid A_3,H_3],\\
  q_3(s_1,A_3,H_3)&=m_3(\max\{A_2,s_1\},H_3),\\
  m_2(s_1,A_2,H_2)&=\E\{q_3(s_1,A_3,H_3)\mid A_2,H_2\},\\
  q_2(s_1,A_2,H_2)&=m_2(s_1,0,H_2),\\
  m_1(s_1,A_1,H_1)&=\E\{q_2(s_1,A_2,H_2)\mid A_1,H_1\},\\
  q_1(A_1,H_1)&=m_1(A_1,0,H_1),\\
  \theta&=\E\{q_1(A_1,H_1)\}.
\end{align*}
Here the unused component $s_2$ is suppressed in $q_3$ and $m_2$.
At time $3$, only the natural value $s_1$ must be remembered in
addition to $A_2$, which is already in $H_3$. Augmentation by the
single column $s_1$ therefore suffices.}

\end{example}

\subsection{Time horizon considerations.}
When natural treatment initiation occurs close to the end of follow-up, the
delayed treatment initiation time may fall outside the study horizon.
For example, in a vaccine study with $\tau=10$ monthly decision times and a five-month
delay, \edit{individuals whose first natural indication for
vaccination generated under the policy occurs in months $6,\ldots,10$} would
remain unvaccinated throughout the observed ten-month horizon under
the delayed intervention. This does not create a problem with our technical framework, but it may affect interpretation because, for
\edit{these late natural indications}, the delayed intervention is 
equivalent within the study horizon to no vaccination. 
%\bibliographystyle{plainnat}
%\bibliography{refs}

{
\section{Asymptotic properties of the estimators}
\label{sec:supp-asymptotic-properties}

Let \(\widehat\theta\) denote either \(\thetatmle\) or
\(\thetasr\). For each validation fold \(j\), define
\(\widehat\eta_j^\star=\widetilde\eta_j\) for the TMLE and
\(\widehat\eta_j^\star=\check\eta_j\) for the SDR estimator.

\begin{theorem}[Consistency, asymptotic normality, and efficiency]
\label{thm:asymptotic-properties}
Assume that the number of folds \(J\) is fixed \edit{and their
sizes satisfy $\min_j n_j/n\geq c>0$ for all sufficiently large $n$}.
Further, assume that:
\begin{enumerate}
\item There exists a constant \(C<\infty\) such that
\[
\max_{1\leq k\leq\tau}
\sup_{\overline s_k\in\overline{\mathcal A}_k}
\norm{\omega_k(\overline s_k)}_{\infty}
\leq C
\]
and
\[
\P\!\left\{
\max_{1\leq j\leq J}
\max_{1\leq k\leq\tau}
\sup_{\overline s_k\in\overline{\mathcal A}_k}
\norm{\widehat\omega_{k,j}(\overline s_k)}_{\infty}
\leq C
\right\}
\longrightarrow 1.
\]
\item
The efficient influence function has finite variance: $
\P\!\left\{
\overline\varphi_1(Z;\eta)^2
\right\}<\infty$.
\item
The estimated influence functions satisfy
$\max_{1\leq j\leq J}
\norm{
\overline\varphi_1(\,\cdot\,;\widehat\eta_j^\star)
-
\overline\varphi_1(\,\cdot\,;\eta)
}_{P,2}
=
o_P(1)$.
\item \edit{For TMLE, assume a fixed horizon and finite treatment
supports. Conditional on the initial training-sample fits,
the class
\[
\mathcal F_{j,n}
=\{\overline\varphi_1(\,\cdot\,;\eta_j(\epsilon)):
\epsilon\in\mathcal E\}
\]
is indexed by the full vector of backward fluctuation parameters.
With probability tending to one, these classes have an envelope
bounded by a fixed $F\in L_2(P)$ and satisfy the uniform entropy bound
\[
\sup_Q N\{u\|F\|_{Q,2},\mathcal F_{j,n},L_2(Q)\}
\leq (B/u)^v,\qquad 0<u<1,
\]
for constants $B,v<\infty$ independent of $j,n$, where the supremum
is over finitely supported probability measures $Q$. The realized
targeted fit belongs to this class. This is a sufficient conditional
Donsker condition on the fluctuation class.}
\end{enumerate}
If
\[
\max_{1\leq j\leq J}|\Rem_0(\widehat\eta_j^\star)|=o_P(1),
\]
then \(\widehat\theta\) is consistent for \(\theta\). If, in addition,
\[
\max_{1\leq j\leq J}|\Rem_0(\widehat\eta_j^\star)|=o_P(n^{-1/2}),
\]
then
\[
\sqrt{n}(\widehat\theta-\theta)
=
\frac{1}{\sqrt{n}}
\sum_{i=1}^n
\overline\varphi_1(Z_i;\eta)
+
o_P(1)
\rightsquigarrow
N\!\left[
0,
\var\{\overline\varphi_1(Z;\eta)\}
\right].
\]
Thus, both estimators are asymptotically normal and attain the nonparametric efficiency bound.
\edit{For bounded outcomes, a fixed time horizon, finite treatment
supports, and the logistic intercept fluctuations used here, the last
condition follows from the usual finite-dimensional VC-type entropy
bounds, together with the bounded-weight condition. Conditional on
training, the regression offsets and weights are fixed functions;
finite sums and products preserve these entropy bounds. No Donsker
condition is imposed on the initial machine-learning regressions.} Moreover, the empirical variance of
$\varphi_1(Z_i;\widehat\eta_{j(i)}^\star)-\widehat\theta$
is consistent for
\(\var\{\overline\varphi_1(Z;\eta)\}\).
\end{theorem}

{
\begin{proof}
Let \(n_j=|V_j|\) denote the number of observations in validation fold \(V_j\), and define
$\pi_{j,n}=\frac{n_j}{n}$. Let \(\P_{n,j}\) denote the empirical measure over \(V_j\). Because
\(V_1,\ldots,V_J\) form a partition of the sample,
\[
\sum_{j=1}^J\pi_{j,n}=1
\qquad\text{and}\qquad
\Pn
=
\sum_{j=1}^J\pi_{j,n}\P_{n,j}.
\]
By assumption, \(J\) is fixed and
\(\min_{1\leq j\leq J}\pi_{j,n}\) is bounded away from zero. Consequently,
\(n_j\asymp n\) uniformly over \(j\). For either estimator, its construction implies
\[
\widehat\theta
=
\sum_{j=1}^J
\pi_{j,n}
\P_{n,j}
\varphi_1(\,\cdot\,;\widehat\eta_j^\star)
+
r_n,
\qquad
r_n=o_P(n^{-1/2}).
\]
For the SDR estimator, this equality follows directly from its definition. For the TMLE, it follows from the cross-fitted efficient influence function estimating equation. If that equation is solved exactly, then \(r_n=0\).

Recall that
$\overline\varphi_1(\,\cdot\,;\eta)
=
\varphi_1(\,\cdot\,;\eta)-\theta$. Subtracting \(\theta\) from the preceding display gives
\[
\widehat\theta-\theta
=
\sum_{j=1}^J
\pi_{j,n}
\P_{n,j}
\overline\varphi_1(\,\cdot\,;\widehat\eta_j^\star)
+
r_n.
\]
Add and subtract
\(\overline\varphi_1(\,\cdot\,;\eta)\) within each fold. Because
$
\sum_{j=1}^J
\pi_{j,n}\P_{n,j}
\overline\varphi_1(\,\cdot\,;\eta)
=
\Pn\overline\varphi_1(\,\cdot\,;\eta)
$
and
$
\P\overline\varphi_1(\,\cdot\,;\eta)=0$,
we obtain
\begin{equation}
\widehat\theta-\theta
=
(\Pn-\P)\overline\varphi_1(\,\cdot\,;\eta)+
\sum_{j=1}^J
\pi_{j,n}\P_{n,j}
\left\{
\overline\varphi_1(\,\cdot\,;\widehat\eta_j^\star)
-
\overline\varphi_1(\,\cdot\,;\eta)
\right\}
+
r_n.
\label{eq:asymptotic-proof-first}
\end{equation}
For each \(j\), define
\[
\Delta_{j,n}
=
\overline\varphi_1(\,\cdot\,;\widehat\eta_j^\star)
-
\overline\varphi_1(\,\cdot\,;\eta).
\]
Decomposing the second term in
\eqref{eq:asymptotic-proof-first} into empirical-process and population components yields
\begin{align}
\widehat\theta-\theta
={}&
(\Pn-\P)\overline\varphi_1(\,\cdot\,;\eta)
\nonumber\\
&+
\sum_{j=1}^J
\pi_{j,n}
(\P_{n,j}-\P)\Delta_{j,n}
+
\sum_{j=1}^J
\pi_{j,n}
\P\Delta_{j,n}
+
r_n.
\label{eq:asymptotic-proof-decomposition}
\end{align}

We consider the last two non-negligible terms separately. First, the conditional von Mises expansion in \edit{Theorem~\ref{theo:sdr}}, evaluated at \(t=0\), gives
\[
\P\overline\varphi_1(\,\cdot\,;\widehat\eta_j^\star)
=
-\Rem_0(\widehat\eta_j^\star).
\]
Since
\(\P\overline\varphi_1(\,\cdot\,;\eta)=0\), it follows that
\[
\P\Delta_{j,n}
=
-\Rem_0(\widehat\eta_j^\star).
\]
Therefore,
\[
\left|
\sum_{j=1}^J
\pi_{j,n}\P\Delta_{j,n}
\right|
\leq
\max_{1\leq j\leq J}
\left|
\Rem_0(\widehat\eta_j^\star)
\right|,
\]
because the \(\pi_{j,n}\) are nonnegative and sum to one.

\edit{Second, consider the empirical-process term. For SDR,
conditional on the training sample, $\Delta_{j,n}$ is fixed and
independent of the observations in $V_j$. Therefore,
\[
\E\left[\{(\P_{n,j}-\P)\Delta_{j,n}\}^2
\mid\mathcal T_j\right]
\leq n_j^{-1}\|\Delta_{j,n}\|_{P,2}^2.
\]
The assumed $L_2(P)$ convergence and conditional Chebyshev's inequality
give $(\P_{n,j}-\P)\Delta_{j,n}=o_P(n^{-1/2})$
\citep[see, e.g.,][]{kennedy2024semiparametric}.

For TMLE, this independence does not hold: the fluctuation parameters
are estimated using validation outcomes, and earlier targeting
updates propagate this dependence backward. Instead, condition only
on the initial training-sample fits and consider the entire
fluctuation-indexed class $\mathcal F_{j,n}$. Its conditional uniform
entropy and envelope conditions imply stochastic equicontinuity.
Together with the assumed $L_2(P)$ convergence, this gives
\[
\sqrt{n_j}(\P_{n,j}-\P)
\{\overline\varphi_1(\,\cdot\,;\widetilde\eta_j)
-\overline\varphi_1(\,\cdot\,;\eta)\}
=o_P(1).
\]
This argument remains valid when the selected fluctuation vector
uses validation data from all folds, because equicontinuity holds
uniformly over that vector \citep{zheng2011cross}.
Thus, for either estimator,
\[
(\P_{n,j}-\P)\Delta_{j,n}=o_P(n^{-1/2}).
\]}
Since \(J\) is fixed and
\(0\leq\pi_{j,n}\leq1\),
\[
\sum_{j=1}^J
\pi_{j,n}
(\P_{n,j}-\P)\Delta_{j,n}
=
o_P(n^{-1/2}).
\]
Substitution into
\eqref{eq:asymptotic-proof-decomposition} gives
\begin{align}
\widehat\theta-\theta
={}&
(\Pn-\P)\overline\varphi_1(\,\cdot\,;\eta)
-
\sum_{j=1}^J
\pi_{j,n}
\Rem_0(\widehat\eta_j^\star)
+
o_P(n^{-1/2}).
\label{eq:asymptotic-proof-master}
\end{align}
If
\[
\max_{1\leq j\leq J}
\left|
\Rem_0(\widehat\eta_j^\star)
\right|
=
o_P(1),
\]
then the second term on the right-hand side of
\eqref{eq:asymptotic-proof-master} is \(o_P(1)\). The first term is
\(O_P(n^{-1/2})\) because
\[
\P\left\{
\overline\varphi_1(Z;\eta)^2
\right\}<\infty.
\]
It follows that $
\widehat\theta-\theta=o_P(1)$, which proves consistency.

Under the stronger condition
\[
\max_{1\leq j\leq J}
\left|
\Rem_0(\widehat\eta_j^\star)
\right|
=
o_P(n^{-1/2}),
\]
the remainder term in
\eqref{eq:asymptotic-proof-master} is \(o_P(n^{-1/2})\). Hence,
\[
\widehat\theta-\theta
=
(\Pn-\P)\overline\varphi_1(\,\cdot\,;\eta)
+
o_P(n^{-1/2}),
\]
or, equivalently,
\[
\sqrt n(\widehat\theta-\theta)
=
\frac{1}{\sqrt n}
\sum_{i=1}^n
\overline\varphi_1(Z_i;\eta)
+
o_P(1).
\]
The classical central limit theorem therefore gives
\[
\sqrt n(\widehat\theta-\theta)
\rightsquigarrow
N\!\left[
0,
\var\{\overline\varphi_1(Z;\eta)\}
\right].
\]
Because \(\overline\varphi_1(Z;\eta)\) is the efficient influence function for \(\theta\) in the nonparametric model, this limiting variance is the nonparametric efficiency bound. Thus, both estimators are asymptotically efficient whenever the stated conditions hold.

It remains to establish consistency of the variance estimator. Define
\[
\widehat D_i
=
\varphi_1\!\left(
Z_i;
\widehat\eta_{j(i)}^\star
\right)
-
\widehat\theta
\]
and
\[
\overline{\widehat D}
=
\frac{1}{n}\sum_{i=1}^n\widehat D_i.
\]
The empirical variance estimator is
\[
\widehat\sigma^2
=
\frac{1}{n}
\sum_{i=1}^n
\left(
\widehat D_i-\overline{\widehat D}
\right)^2.
\]
\edit{For SDR, conditional Markov's inequality and the assumed
$L_2(P)$ convergence give $\P_{n,j}\Delta_{j,n}^2=o_P(1)$.
For TMLE, the conditional entropy and square-integrable envelope
conditions give a uniform law of large numbers for the corresponding
class of squared differences. Hence
$(\P_{n,j}-\P)\Delta_{j,n}^2=o_P(1)$ and again
$\P_{n,j}\Delta_{j,n}^2=o_P(1)$. Combining this with consistency of
$\widehat\theta$ gives}
\[
\frac{1}{n}
\sum_{i=1}^n
\left[
\widehat D_i
-
\overline\varphi_1(Z_i;\eta)
\right]^2
=
o_P(1).
\]
Consequently,
\[
\frac{1}{n}\sum_{i=1}^n\widehat D_i^2
=
\frac{1}{n}
\sum_{i=1}^n
\overline\varphi_1(Z_i;\eta)^2
+
o_P(1)
\]
and
\[
\overline{\widehat D}=o_P(1).
\]
By the law of large numbers,
\[
\frac{1}{n}
\sum_{i=1}^n
\overline\varphi_1(Z_i;\eta)^2
\longrightarrow_P
\P\!\left\{
\overline\varphi_1(Z;\eta)^2
\right\}.
\]
Because
\(\P\overline\varphi_1(Z;\eta)=0\), we conclude that
\[
\widehat\sigma^2
\longrightarrow_P
\var\{\overline\varphi_1(Z;\eta)\}.
\]
This completes the proof.
\end{proof}
}

For completeness, the remainder condition can be expressed using familiar product-rate conditions. Define
\[
r_{g,t,j}
=
\max_{s_t\in\mathcal A_t}
\norm{
\frac{\widehat g_{t,j}(s_t\mid H_t)}
     {\widehat g_{t,j}(A_t\mid H_t)}
-
\frac{g_t(s_t\mid H_t)}
     {g_t(A_t\mid H_t)}
}_{P,2}
\]
and
\[
r_{m,t,j}
=
\max_{\overline s_t\in\overline{\mathcal A}_t}
\norm{
\widehat m_{t,j}^\star(\overline s_t,A_t,H_t)
-
m_t(\overline s_t,A_t,H_t)
}_{P,2}.
\]
\edit{For a fixed horizon and fixed finite treatment supports, the boundedness
conditions in the theorem, the second-order remainder in
Theorem~\ref{theo:sdr}, and the Cauchy--Schwarz inequality yield}
\[
\left|
\Rem_0(\widehat\eta_j^\star)
\right|
\lesssim
\sum_{t=1}^{\tau}
r_{g,t,j}r_{m,t,j}.
\]
Consequently, a sufficient condition for asymptotic linearity is
\[
\max_{1\leq j\leq J}
\sum_{t=1}^{\tau}
r_{g,t,j}r_{m,t,j}
=
o_P(n^{-1/2}).
\]
For example, this condition holds if, at every time point, the relevant treatment and outcome-regression estimators converge at rates whose product is \(o_P(n^{-1/2})\); the familiar \(o_P(n^{-1/4})\) rate for both nuisance estimators is one sufficient special case.
}

{
  \section{Extension to right-censored
    data}\label{sec:right-censoring}

Let \(C_t=1\) indicate that an individual remains uncensored through
the end of interval \(t\) (i.e., the individual remains under
observation at \(t+1\)), and assume monotone censoring. Thus,
\(C_t=0\) implies \(C_k=0\) for every \(k>t\), and variables occurring
after censoring \(C_t=0\) are not observed. The counterfactual
parameter of interest under a G-LMTP \(d\) and elimination of
censoring is
\[
\theta_C^d
=
E\{Y(\bar A^d,\bar C=\bar 1)\}.
\]
The contrast considered in the application is
\[
\theta_C^d-\theta_C^{\mathrm{id}}
=
E\{Y(\bar A^d,\bar C=\bar 1)\}
-
E\{Y(\bar C=\bar 1)\},
\]
where \(\mathrm{id}\) denotes the policy that leaves the treatment
process unchanged.

Right censoring can be incorporated into the framework by defining the
joint intervention variable $B_t=(A_t,C_t)$ and the policy
\[
d_t^B((\bar s_t,\bar c_t),h_t)
=\{d_t(\bar s_t,h_t),1\},
\]
where $\bar s_t$ and $\bar c_t$ collect the natural treatment and
censoring inputs, respectively. Thus, the first
component is assigned according to the G-LMTP and the second component
is assigned according to the static intervention that prevents
censoring. This construction is the standard treatment of right
censoring in the LMTP framework
\citep{diaz2023nonparametric,diaz2024survival}. Define
\[
g_{A,t}(a_t\mid h_t)
=
P(A_t=a_t\mid H_t=h_t,\bar C_{t-1}=\bar 1)
\]
and
\[
g_{C,t}(1\mid a_t,h_t)
=
P(C_t=1\mid A_t=a_t,H_t=h_t,\bar C_{t-1}=\bar 1).
\]
The joint intervention mechanism factors as
\[
g_{B,t}(a_t,c_t\mid h_t)
=
g_{A,t}(a_t\mid h_t)
g_{C,t}(c_t\mid a_t,h_t),
\]
where conditioning on \(\bar C_{t-1}=\bar 1\) is suppressed in the
notation on the right-hand side.

In addition to consistency and the treatment positivity and strong
sequential randomization assumptions in Section~3, applied among
individuals satisfying \(\bar C_{t-1}=\bar 1\), identification requires
the following conditions.

\paragraph{Censoring positivity.}
For every treatment and covariate history having positive probability
under the intervention, $ g_{C,t}(1\mid a_t,h_t)>0$ for
$t=1,\ldots,\tau$.

\paragraph{Sequential randomization of censoring.}
Conditional on the measured history immediately before censoring,
censoring at time \(t\) is independent of the future counterfactual
variables under the interventions under consideration. In particular,
for every supported treatment regimen $\bar a$ (a fixed sequence of
assigned treatments, distinct from the natural input sequence $\bar s$),
\[
C_t
\perp\!\!\!\perp
\left\{
L_k(\bar a,\bar 1),
A_k(\bar a,\bar 1),
Y(\bar a,\bar 1):
k>t
\right\}
\ \big|\ 
A_t,H_t,\bar C_{t-1}=\bar 1.
\]
This condition requires \(H_t\) to contain the common causes of
censoring at time \(t\) and the subsequent treatment, covariate, and
outcome processes. The sequential regression representation can be stated directly.
First define
\[
m_\tau^C(a_\tau,h_\tau)
=
E\left(
Y
\mid
A_\tau=a_\tau,H_\tau=h_\tau,\bar C_\tau=\bar 1
\right)
\]
and
\[
q_\tau^C(\bar s_{\tau-1},a_\tau,h_\tau)
=
m_\tau^C
\left\{
d_\tau((\bar s_{\tau-1},a_\tau),h_\tau),
h_\tau
\right\}.
\]
Recursively, for \(t=\tau-1,\ldots,1\), define
\[
m_t^C(\bar s_t,a_t,h_t)
=
E\left\{
q_{t+1}^C(\bar s_t,A_{t+1},H_{t+1})
\mid
A_t=a_t,H_t=h_t,\bar C_t=\bar 1
\right\}
\]
and
\[
q_t^C(\bar s_{t-1},a_t,h_t)
=
m_t^C
\left\{
(\bar s_{t-1},a_t),
d_t((\bar s_{t-1},a_t),h_t),
h_t
\right\}.
\]
Under the assumptions above,
\[
\theta_C^d
=
E\{q_1^C(A_1,H_1)\}.
\]
This result follows by applying Theorem~1 to the augmented intervention
process \(B_t=(A_t,C_t)\), or equivalently by applying the extended
g-computation formula to the treatment policy together with the static
intervention \(\bar C=\bar 1\).

The TMLE and SDR in Section~5 are applied to this augmented intervention
process. Operationally, the outcome regressions at time \(t\) are
estimated among individuals satisfying \(\bar C_t=\bar 1\), and every
treatment-mechanism factor
\[
\frac{g_{A,t}(s_t\mid H_t)}
     {g_{A,t}(A_t\mid H_t)}
\]
appearing in the estimating functions is replaced by
\[
\frac{
g_{A,t}(s_t\mid H_t,\bar C_{t-1}=\bar 1)
}{
g_{A,t}(A_t\mid H_t,\bar C_{t-1}=\bar 1)
}
\frac{
C_t
}{
g_{C,t}(1\mid A_t,H_t,\bar C_{t-1}=\bar 1)
}.
\]
Thus, censoring contributes the usual inverse probability of remaining
uncensored. For the reference parameter
\(E\{Y(\bar C=\bar 1)\}\), the treatment policy is the identity, so the
treatment component of this ratio cancels and only the censoring
component remains. The efficient influence function, TMLE and SDR
algorithms, and asymptotic results in Sections~4--5 consequently apply
with \(g_{B,t}\) in place of \(g_t\). In particular, the stated
boundedness and product-rate conditions apply to the augmented
treatment-and-censoring mechanism.
}

\end{document}